\documentclass[acmsmall,nonacm]{acmart}
\AtBeginDocument{%
  }

\acmYear{2025}
\acmJournal{PACMPL}

\usepackage{color, xcolor}
\usepackage{soul}
\usepackage{verbatim}

\usepackage[utf8]{inputenc}
\usepackage{booktabs}
\usepackage{multirow}
\usepackage{graphicx}
\usepackage{makecell}
\usepackage{tabularx}

\usepackage{algorithm}
\usepackage{algorithmicx}
\usepackage{algpseudocode}
\usepackage{wasysym}
\usepackage{caption} 

\usepackage{algorithm}
\usepackage{algpseudocode}
\usepackage{caption}
\usepackage{graphicx}

\usepackage{pifont}

\let\oldding\ding
\renewcommand{\ding}[2][1]{\scalebox{#1}{\oldding{#2}}}

\newcommand{\stitle}[1]{\vspace{0.5ex}\noindent{\bf #1}}
\newcommand{\etitle}[1]{\vspace{0.5ex}\noindent{\bf \textit{#1}}}

\begin{document}

\title{Stencil-Lifting: Hierarchical Recursive Lifting System for Extracting Summary of Stencil Kernel in Legacy Codes}

\subtitle{Technical Report}

\titlenote{This is an extended version of the OOPSLA2'25 publication (DOI: 10.1145/3763159). }

\author{Mingyi Li}
\orcid{0000-0003-4266-908X}
\affiliation{%
  \institution{Institute of Computing Technology, Chinese Academy of Sciences}
  \city{Beijing}
  \country{China}
}
\email{limingyi@ncic.ac.cn}

\author{Junmin Xiao}
\orcid{0000-0003-0457-4709}
\affiliation{%
  \institution{Institute of Computing Technology, Chinese Academy of Sciences}
  \city{Beijing}
  \country{China}
}
\email{xiaojunmin@ict.ac.cn}

\author{Siyan Chen}
\orcid{0009-0008-3778-5995}
\affiliation{%
  \institution{Institute of Computing Technology, Chinese Academy of Sciences}
  \city{Beijing}
  \country{China}
}

\author{Hui Ma}
\orcid{0000-0002-7083-9552}
\affiliation{%
  \institution{Institute of Computing Technology, Chinese Academy of Sciences}
  \city{Beijing}
  \country{China}
}

\author{Xi Chen}
\orcid{0009-0008-4571-0213}
\affiliation{%
  \institution{Institute of Computing Technology, Chinese Academy of Sciences}
  \city{Beijing}
  \country{China}
}

\author{Peihua Bao}
\orcid{0009-0005-7264-9102}
\affiliation{%
  \institution{University of Chinese Academy of Sciences}
  \city{Beijing}
  \country{China}
}

\author{Liang Yuan}
\orcid{0000-0003-3406-2907}
\affiliation{%
  \institution{Institute of Computing Technology, Chinese Academy of Sciences}
  \city{Beijing}
  \country{China}
}

\author{Guangming Tan}
\orcid{0000-0002-6361-5948}
\affiliation{%
  \institution{Institute of Computing Technology, Chinese Academy of Sciences}
  \city{Beijing}
  \country{China}
}
\email{tgm@ict.ac.cn}

\begin{abstract}
We introduce Stencil-Lifting, a novel system for automatically converting stencil kernels written in low-level languages in legacy code into semantically equivalent Domain-Specific Language (DSL) implementations. Targeting the efficiency bottlenecks of existing verified lifting systems, Stencil-Lifting achieves scalable stencil kernel abstraction through two key innovations. First, we propose a hierarchical recursive lifting theory that represents stencil kernels, structured as nested loops, using invariant subgraphs, which are customized data dependency graphs that capture loop-carried computation and structural invariants. Each vertex in the invariant subgraph is associated with a predicate-based summary, encoding its computational semantics. By enforcing self-consistency across these summaries, Stencil-Lifting ensures the derivation of correct loop invariants and postconditions for nested loops, eliminating the need for external verification. Second, we develop a hierarchical recursive lifting algorithm that guarantees termination through a convergent recursive process, avoiding the inefficiencies of search-based synthesis. The algorithm efficiently derives the valid summaries of stencil kernels, and its completeness is formally proven. We evaluate Stencil-Lifting on diverse stencil benchmarks from two different suites and on four real-world applications. Experimental results demonstrate that Stencil-Lifting achieves 31.62$\times$ and 5.8$\times$ speedups compared to the state-of-the-art verified lifting systems STNG and Dexter, respectively, while maintaining full semantic equivalence. Our work significantly enhances the translation efficiency of low-level stencil kernels to DSL implementations, effectively bridging the gap between legacy optimization techniques and modern DSL-based paradigms.
\end{abstract}



\maketitle

\section{Introduction}

Stencil computation is a fundamental kernel in scientific computing \cite{dwarf06, dwarf09} and is widely used in numerical simulations of physical processes. Over the past decade, extensive efforts have been devoted to optimizing stencil computations on heterogeneous accelerators, which offer significant performance potential \cite{Maruyama2011, Liu2022}.


Among various optimization approaches, domain-specific languages (DSLs) for stencil computation have played a central role. By decoupling computational logic from optimization strategies, DSLs provide a flexible and effective optimization framework \cite{Ragan-Kelley2013, Kukreja2016}. Studies show that starting from high-level DSL code can yield highly efficient implementations \cite{Catanzaro2010, Tang2011, Ragan-Kelley2013, Mendis2015, Kukreja2016, Ahmad2019, Baghdadi2019, Brown2022}, which often outperform those generated by state-of-the-art compilers on manually optimized C or Fortran codes. High-level DSLs free the compiler from low-level constraints, enabling more systematic and efficient optimization than manual tuning.

However, applying DSLs to legacy code is challenging. These codes, typically written in low-level languages with platform-specific optimizations, are hard to read, maintain, and port. Once original developers leave, adapting them to new architectures becomes error-prone and labor-intensive.

To leverage DSL optimizations for legacy code without manual rewriting, researchers propose automatic conversion approaches \cite{Mendis2015, Kamil2016, Ahmad2018, Ahmad2019}. These transform low-level stencil code into DSLs, enabling parallelization, cache tuning, and vectorization. This transformation not only enhances maintainability but also enables efficient execution on heterogeneous architectures, such as GPUs.

Automatic code conversion techniques are classified into two categories: syntax-guided pattern matching \cite{Ginsbach2018, Espindola2023, VanderCruysse2024} and verified lifting \cite{Kamil2016, Ahmad2019, Laddad2022, Verbeek2022, Qiu2024}. Pattern matching relies on syntax rules to convert code into DSLs, but struggles with limited pattern coverage and a lack of semantic guarantees. In contrast, verified lifting uses synthesis to derive a formally verified summary, which is an abstract predicate representation of stencil computation, then translates it into DSLs. Based on synthesis and SMT solvers, it provides strong guarantees of correctness. However, its reliance on search-based strategies means that valid summaries may not always be found within a bounded time. STNG \cite{Kamil2016} and Dexter \cite{Ahmad2019} exemplify this approach, targeting Fortran-to-Halide and C++ image processing, respectively. Both adopt a Counterexample-guided Inductive Synthesis (CEGIS) \cite{Berlin2009} framework to derive summaries and verify correctness using Z3 \cite{DeMoura2008}.

In practical scenarios, stencil kernels mainly involve array-level updates within nested loops, which complicates direct loop summarization, or \textit{lifting} as mentioned earlier. The main difficulties arise from two factors: (1) array updates require a unified expression capturing element-wise modifications across the target array, making summarization hard; (2) nested loops introduce complex multi-dimensional data dependencies, further hindering abstraction and analysis.

Traditional analyses, such as symbolic execution and dataflow analysis, work well for loops that update scalars, utilizing symbolic evaluation or recurrence solving. However, with arrays, conventional methods often resort to loop unrolling, failing to abstract inter-element dependencies. This limits recursive analysis for summarizing nested loops with array transformations, hindering robust summarization for realistic stencil computations.

To achieve effective \textit{lifting}, we investigate stencil dataflows. Stencil computations perform nearest-neighbor updates on multidimensional grids, where each point is updated based on itself and neighbors. This reveals recurring invariant structures in the data dependency graph. To capture these structures, we introduce the concept of an invariant subgraph at each loop level and treat iteration indices symbolically, enabling symbolic execution to derive uniform update expressions.

Moreover, each update uses a fixed set of neighbors, independent of loop index. This indicates that the inner loop performs a fixed computation pattern, and is repeatedly applied under different outer-loop contexts. To capture this invariant, we encapsulate the inner loop as a function updating a range of elements in the output array. This functional abstraction supports recursive decomposition of higher-level loops, enabling structured level-by-level summarization.


Building on these insights, we present \textit{Stencil-Lifting}, a novel system for hierarchical summarization of stencil kernels in legacy codes. Our approach introduces two key innovations that distinguish it from prior work: First, we establish a hierarchical recursive lifting theory where each vertex in the invariant subgraph carries a predicate-based summary encoding its semantics (Section~\ref{section: Lifting Theory}). By constructing self-consistent vertex summaries, loop invariants and postconditions are derived directly, with a consistency proof inherently ensuring semantic equivalence (obviating external verification). Second, we propose a provably terminating recursive lifting algorithm that substitutes traditional search-based synthesis with a convergent recursive process (Section~\ref{section: Summary Lifting}). Its termination property guarantees sound summarization for stencil kernels, ensuring the algorithm's completeness. Specifically, our contributions can be summarized as follows:

\stitle{\small \CIRCLE}  Establish a novel system Stencil-Lifting to automatically translate Fortran stencil kernels in legacy codes to the semantically equivalent Halide DSL (Section \ref{System Overview}).

\stitle{\small \CIRCLE}  Propose a new summary lifting theory based on the data dependency graph. In the theory, the concept of an invariant graph is presented. The stencil summary lifting is proved to be equivalent to searching for a \textit{self-consistent} summary configuration in the entire invariant graph (Section \ref{section: Lifting Theory}).

\stitle{\small \CIRCLE}  Develop a hierarchical recursive lifting algorithm guiding all vertex summaries in an invariant subgraph towards self-consistency, while formally proving the algorithm completeness (Section \ref{section: Summary Lifting}).

\stitle{\small \CIRCLE} Present a series of experiments to evaluate Stencil-Lifting by translating typical stencil benchmarks from two different suites and four real-world applications, which demonstrate the effectiveness and efficiency of Stencil-Lifting (Section \ref{section: Evaluation}).

\begin{figure}[htbp]
    \centering
    \includegraphics[width=0.75\textwidth]{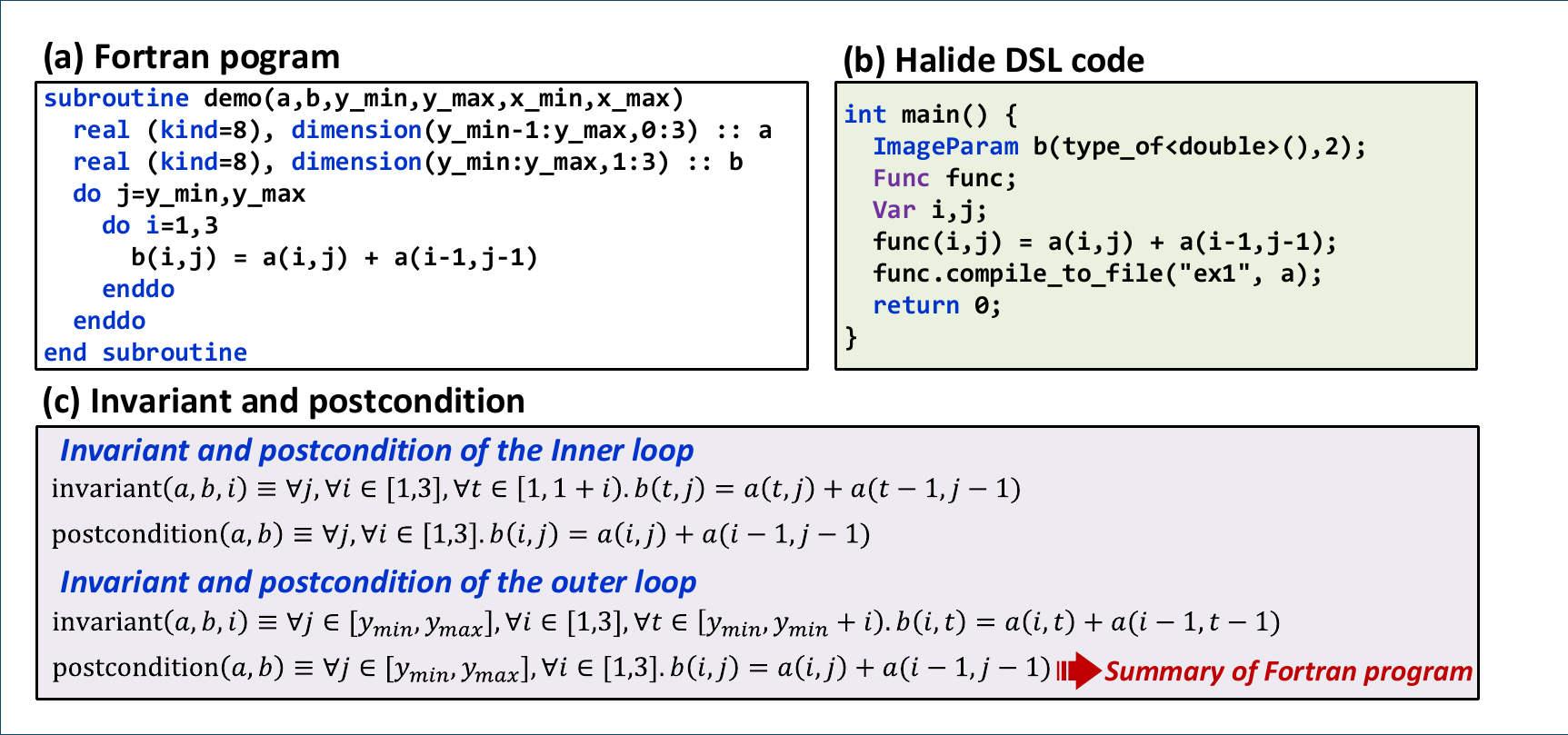}
    \vspace{-1.0em}
    \caption{A simple example of stencils. Panel (a) presents the original Fortran program. Panel (b) illustrates the corresponding Halide DSL code. Panel (c) shows the invariants and postconditions of inner and outer loops.}
    \label{fig:lifting-illus}
\end{figure}

\section{Background}
\subsection{Stencil Computation and Optimization}

Stencil computation consists of loops that update array elements using fixed access patterns, forming the basis of many scientific applications \cite{dwarf06, dwarf09, Gysi2015, Singh2020}. Stencils differ by grid dimensionality (1D–3D), neighborhood size (e.g., 3-point, 5-point), shape (box, star), and boundary conditions (constant, periodic). Despite their ubiquity, they suffer from low arithmetic intensity, poor locality, and complex dependencies, making optimization a long-standing challenge. DSL-based approaches address this by separating computation from optimization. For example, Halide \cite{Ragan-Kelley2013} decouples algorithm specification from scheduling, enabling vectorization and memory optimizations without altering the algorithm.

\subsection{Stencil Verified Lifting}

To enable DSL-level optimizations for legacy code without the burden and potential risk of manual rewriting, various approaches have been proposed to automatically convert low-level stencil code into DSLs. Throughout the transformation, the key challenge lies in ensuring semantic equivalence.

Over the past decade, verified lifting has emerged as a prominent methodology \cite{Kamil2016, Ahmad2019, Polygeist2021}. It derives a mathematical summary of the loop nest, capturing the effect of stencil execution on output arrays. This summary, typically expressed in a high-level predicate language (Figure~\ref{fig:lifting-illus}), is then translated into the target DSL (e.g., Halide). If the summary is verified to match the original semantics, the process qualifies as verified lifting.

Verifying the correctness of the summary can be framed in Hoare logic, where the summary serves as a postcondition (a predicate that holds after all executions of the code, assuming certain preconditions). Ensuring its validity requires constructing loop invariants (predicate expressions that remain true before and after each loop iteration, regardless of loop index or iteration count). Figure~\ref{fig:lifting-illus}(c) illustrates loop invariants and postconditions for both inner and outer loops. Their correctness is validated by satisfying the three propositions in Equation (\ref{equation: VC}) for each loop level:
\begin{equation}
\label{equation: VC}
\begin{array}{l}
  (I) \ \forall s \, \big( \mathit{pre}(s) \to \mathit{invar}(s) \big), \\
  (II) \ \forall s \, \big( \mathit{invar}(s) \land \mathit{cond}(s) \to \mathit{invar}(\text{body}(s)) \big), \\
  (III) \ \forall s \, \big( \mathit{invar}(s) \land \neg \mathit{cond}(s) \to \mathit{post}(s) \big).
\end{array}
\end{equation}
These propositions indicate that (I) if the precondition is true for any state $s$, then the invariant is also true, (II) if the invariant holds and the loop condition is true, then the invariant should hold for the state obtained after executing the $body$ of the loop once on $s$, and (III) for all states where the invariant holds but which do not satisfy the loop condition, the postcondition should hold as well.

\section{Challenges and Motivation}

\subsection{Challenges for Stencil Lifting}

A stencil kernel typically involves array (rather than scalar) updates and nested loops, complicating direct summary lifting.

\stitle{Challenge 1: The efficiency barrier lies in deriving unified array update expressions for stencil summarization.} Existing methods use symbolic execution to generate initial summary templates. For example, in a 1d-3p stencil, symbolic execution may produce expressions for individual output elements, such as $B[3]=A[2]+A[3]+A[4]$ and $B[6]=A[5]+A[6]+A[7]$, but fails to generalize them into parameterized forms like $B[i]=A[i-1]+A[i]+A[i+1]$. Instead, a summary starts is initially abstracted as $B[*]=A[*]+A[*]+A[*]$, and a program synthesizer searches for index candidates, which are iteratively refined and validated through a CEGIS process. Summaries are confirmed only after all candidates pass SMT-based verification, as in STNG. While effective for low-dimensional stencils, this approach faces two scalability issues: (1) The index search space grows exponentially with stencil dimensionality. (2) Verification becomes costly as counterexample-guided refinement for complex stencils generates an exponentially growing number of candidates, leading to solver timeouts. Figure~\ref{fig:challenge-2}(a) shows that lifting a 3d-19p stencil takes $7\times$ longer than a 3d-7p variant, and 3d-27p case fails to complete within 18 hours due to search space explosion. The fundamental barrier is the inability to reason about array accesses at the symbolic level, suggesting current methods lack high-level index abstraction to capture and exploit repetitive patterns.

\begin{figure}
    \centering
    \includegraphics[width=0.97\linewidth]{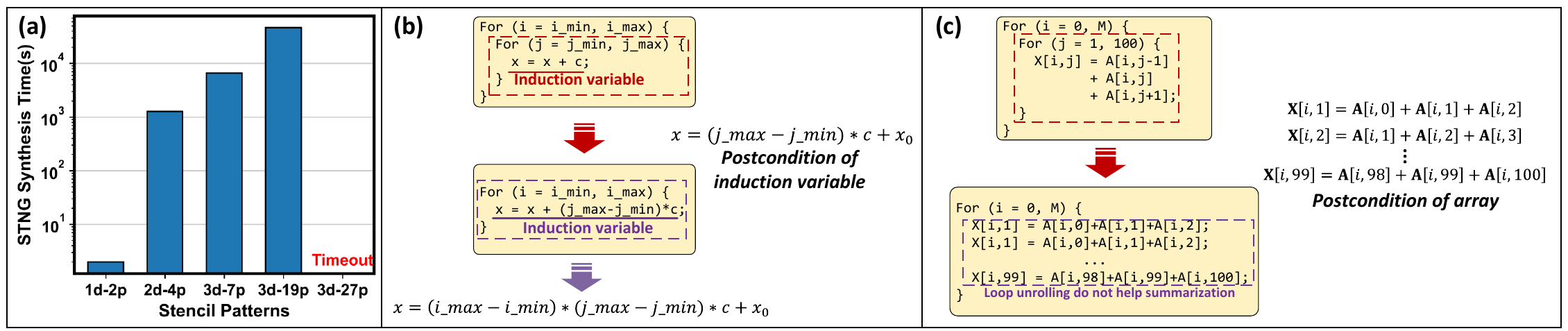}
    \vspace{-0.2cm}
    \caption{Summarization. Plane (a) presents STNG time for summary lifting. Plane (b) shows the loop summarization for scalar computation. Plane (c) illustrates the loop summarization for array computation.}
    \label{fig:challenge-2}
\end{figure} 

\stitle{Challenge 2: Cross-layer data dependencies complicate dataflow loop summarization.} Cross-layer data dependencies in nested loops pose challenges for dataflow-based summarization, especially in array-intensive computations. Existing methods typically summarize the innermost loop, convert it into a linear program, and reintegrate it to eliminate nesting~\cite{zhu2024, S-looper2015, Proteus2016, Blicha2022}. These approaches work well for scalar updates, where inner loops reduce to closed-form equations (Figure~\ref{fig:challenge-2}(b)). However, they often fail on arrays, yielding verbose, element-wise results akin to full loop unrolling, without capturing cross-iteration dependencies. The core difficulty lies in abstracting loop-carried dependencies from complex memory access patterns, which hierarchical summarization struggles to model. This becomes more pronounced in real-world applications, where performance-oriented transformations (e.g., tiling, fusion) further complicate iteration spaces and dependencies. These transformations introduce dynamic memory-access patterns and loop bounds, making static summarization infeasible.

\subsection{Motivation}

Our work is motivated by the two observations below:

\stitle{Observation 1: The fixed computation pattern reveals an invariant structure in the dataflow graph.}
Stencil computations apply a fixed update pattern to all elements in the output array. For instance, a 2d-2p stencil updates elements via $B(i,j) = A(i,j) + A(i-1,j-1)$ (Figure~\ref{fig:lifting-illus}). More generally, for a $d$-dimensional $m$-point stencil with input array $A$ and output array $B$, the computation pattern can be expressed as:
\begin{equation}
\label{equation: stencil computation pattern}
B[i_1, \dots, i_d] = \lambda_1 \cdot A[t_1(i_1, \dots, i_d)] + \lambda_2 \cdot A[t_2(i_1, i_2, \cdots, i_d)] + \cdots + \lambda_m \cdot A[t_m(i_1, \dots, i_d)],
\end{equation}
where each $t_k$ is an affine transformation and $\lambda_k$ a scalar coefficient. Symbolic execution can directly determine ${\lambda_1, \dots, \lambda_m}$, but existing methods derive ${t_1, \dots, t_m}$ via inductive synthesis. Equation~(\ref{equation: stencil computation pattern}) not only serves as a foundation for inferring $t_1, t_2, \dots, t_m$, but also reveals a key invariant structure in the stencil’s data dependence graph: since $m$ remains constant, each output element is derived from an identical structural pattern. Based on this observation, we propose the novel concept of an \textit{invariant subgraph}, which is isomorphic to every recurring stencil pattern. By treating array indices as symbolic variables instead of concrete values, symbolic execution within this subgraph can generalize the update expression for all array elements.

\stitle{Observation 2: Hierarchical functionalization enables recursive stencil lifting.}
Each update of $B$ in a stencil computation involves exactly $m$ elements from $A$, independent of the loop index. Thus, the inner loop executes a fixed computation scheme repeatedly for different outer loops. For example, in a 2d-2p stencil (Figure~(\ref{fig:lifting-illus})), when the outer loop index $j$ takes values 1 and 5, the inner loop over $i$ performs the same update for the first and fifth rows of $B$. To capture this invariant inner-loop pattern, we functionalize the inner loop as a high-level function that updates a specified range of $B$ in one call, rather than treating each element update in isolation. By representing this function as a vertex in the data dependence graph of the outer loop, we enable recursive functionalization across loop levels. This hierarchical abstraction allows stencil lifting to proceed recursively from inner to outer loops, establishing an efficient algorithm for summary lifting.

\begin{figure}[btph]
    \centering
    \includegraphics[width=0.9\textwidth]{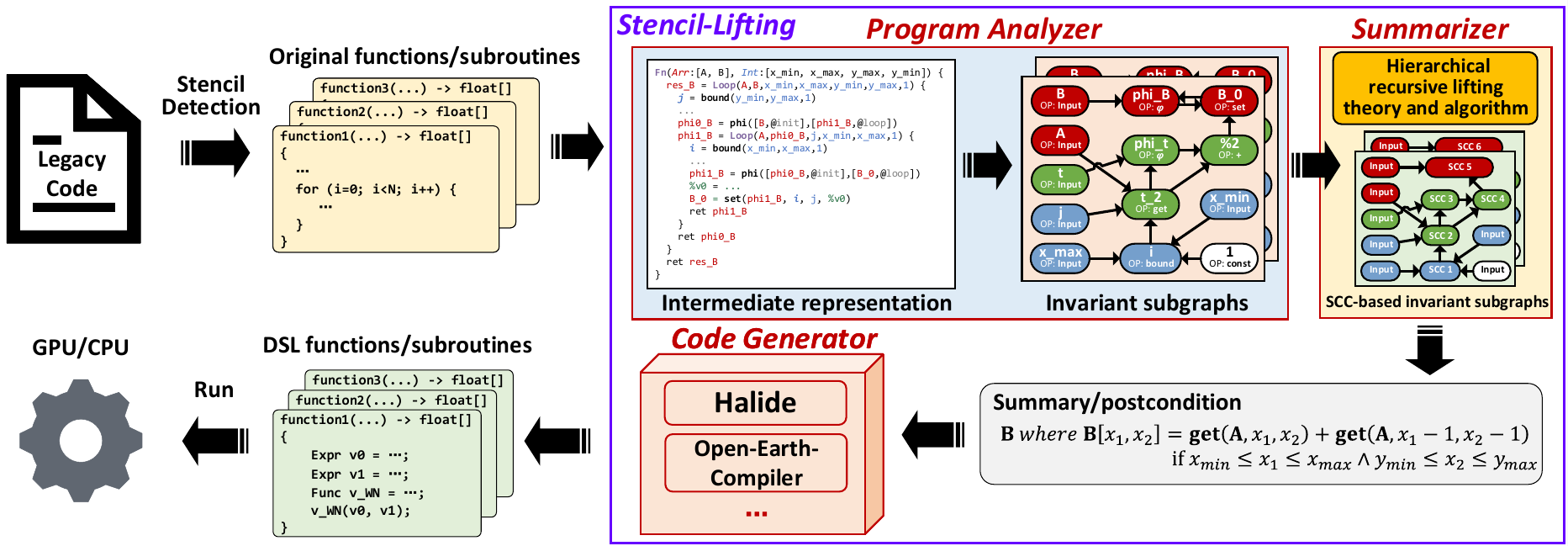}
    \vspace{-0.5em}
    \caption{Overview of Stencil-Lifting.}
    \label{fig:system-overview}
    \vspace{-1.5em}
\end{figure}

\section{System Overview}
\label{System Overview}

Stencil-Lifting is a recursive summary-lifting system that analyzes legacy Fortran code to extract high-level predicate summaries of stencil kernels, which are then translated into DSLs like Halide for high-performance code generation. The system comprises three core modules (Figure \ref{fig:system-overview}):

\stitle{\small \CIRCLE} Program Analyzer transforms Fortran stencil kernels into invariant subgraphs that encode loop structures. It parses source code into an intermediate representation (IR), performs static analysis to extract control flow and dependencies, and identifies invariant computation patterns within each loop. These are represented as invariant subgraphs (directed, possibly cyclic graph representing intra-loop behavior) and passed to the recursive lifting stage. Inner loops are treated as functions invoked by outer loops, enabling hierarchical functionalization.

\stitle{\small \CIRCLE} Summarizer lifts stencil summaries from invariant subgraphs using a hierarchical recursive lifting algorithm. Specifically, it infers vertex summaries that capture the computational semantics of each vertex. According to the proposed hierarchical recursive lifting theory, this process is equivalent to determining all vertex summaries under a set of predefined summary rules. By exploiting the hierarchical dependency relationships between invariant subgraphs, Summarizer recursively traverses them and applies a semantic summarization in each invariant subgraph. Ultimately, the extracted vertex summaries form the stencil’s postcondition, serving as its final summary.

\stitle{\small \CIRCLE} Code Generator translates the lifted summary into DSL code, targeting platforms such as GPUs and many-core architectures. Due to the structural similarity between our predicate language and DSLs (e.g., Halide), translation is straightforward. The generated DSL code can further benefit from built-in autotuning to explore optimized schedules for diverse hardware backends.

\etitle{Remark.}
Stencil-Lifting and STNG/Dexter differ fundamentally in design philosophy. STNG \cite{Kamil2016}, enhanced by Dexter \cite{Ahmad2019} with specialized generators, adopts a guess-and-check approach based on Sketch \cite{Berlin2009}, repeatedly generating and verifying candidate summaries. While effective, this incurs significant computational cost. In contrast, Stencil-Lifting leverages a recursive analysis framework that extracts invariant structures and propagates summaries across loop hierarchies. This design avoids exhaustive validation, reducing complexity.

\section{Hierarchical Recursive Lifting Theory of Stencil Summary}
\label{section: Lifting Theory}

\subsection{Graph-Based Abstraction of Stencil Programs}
\label{subsection: Graph-Based Abstraction}

\subsubsection{ Formalization of Loops in Stencil Code}
\label{section: Formalization of Loops in Stencil Code}
To enable program analysis and transformation, Stencil-Lifting focuses on stencil kernels structured as nested loops.

\begin{definition}[$K$-level Nested Loop]
\label{definition: Nested Loop}
A loop $L: \text{Loop}(Init, Cond, Step, Body)$ is a program construct. Here, $Init$ is a set including the initialized scalar loop indices (e.g., $t$) and arrays (e.g., $A$, $B$); $Cond$ defines the loop conidtion (e.g., $t < T$); $Step$ gives the loop index update way (e.g., $t = t + 1$); $Body$ is the loop body of $L$, which is represented as a sequence of statements. Let $\text{Statement}(Body)$ be the set of all statements in $Body$. Each statement represents one of four basic computation semantics, i.e., assignment, arithmetic operations, conditionals, and function calls.

If another loop $L'$ satisfies $\text{Statement}(Body~of~L') \subset \text{Statement}(Body~of~L)$, we say that $L'$ is nested within loop $L$ (denoted $L' \subset L$). More generally, a $K$-level nested loop is defined recursively:\\
(1) A $1$-level loop is a loop with no other loop in its body.
\\
(2) A $K$-level nested loop is a loop that contains at least one $(K{-}1)$-level nested loop in its body.
\end{definition}

Based on this recursive definition, the innermost loop is called the first level loop $L^{(1)}$, and the outermost loop is the $K$-th level loop $L^{(K)}$. Other level loops are denoted by $L^{(k)}$ sequentially.

\etitle{Remark.} To the best of our knowledge, the stencil kernel is usually programmed as a $K$-level nested loop $L$ including loops $L^{(1)}, L^{(2)}, \cdots, L^{(K)}$, where $L^{(k)} \subset L^{(k+1)}$.

\begin{figure}
    \centering
    \includegraphics[width=0.92\linewidth]{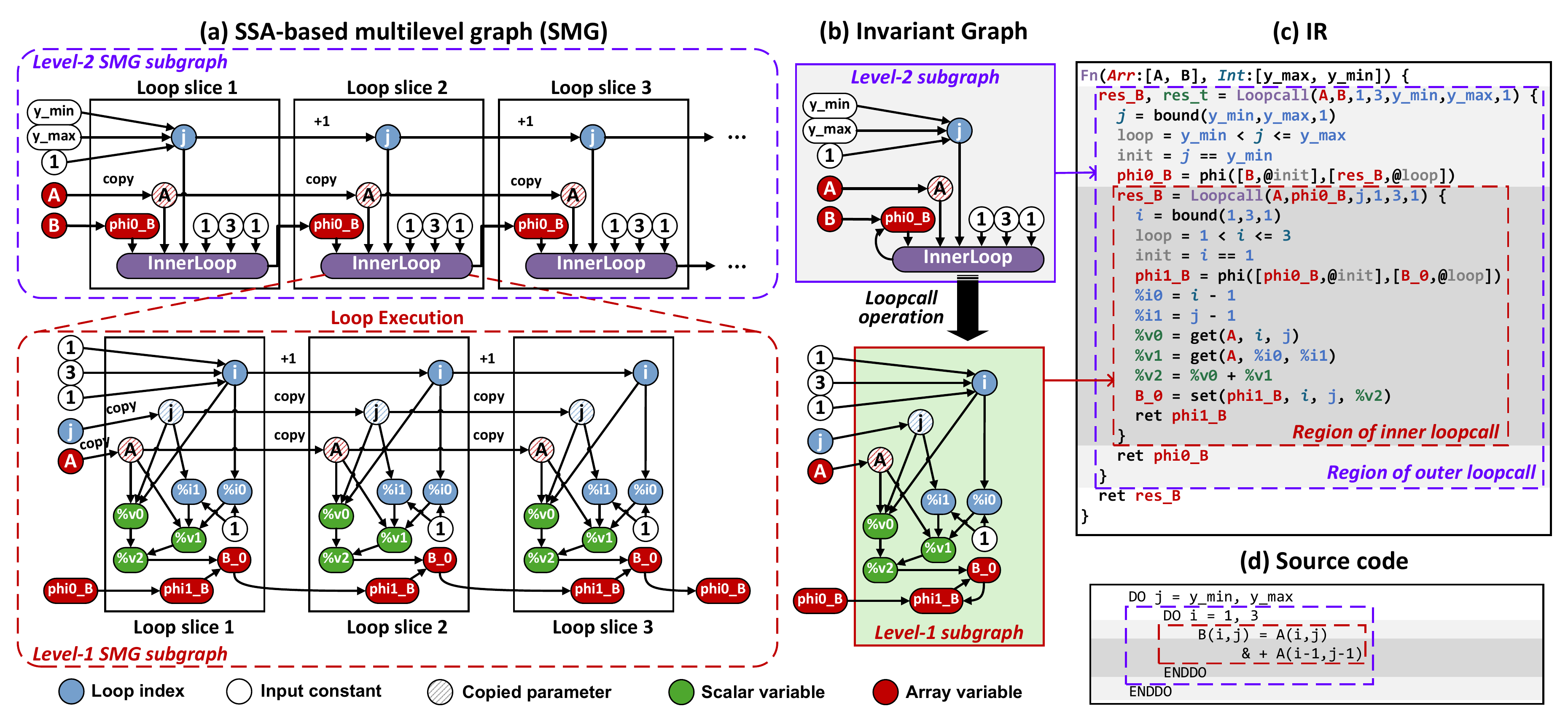}
    \vspace{-0.2em}
    \caption{An example of SSA-based multilevel graph, invariant graph, and IR constructed from a Fortran code.}
    \label{fig:motivation-2}
    \vspace{-0.2cm}
\end{figure}

\subsubsection{ Multilevel Data Dependency Graph}
\label{section: Multilevel Data Dependency Graph}
To support an effective summary lifting process, we model the $K$-level nested loop $L$ as a multilevel data dependency graph, where the subgraph at level $k$ corresponds to the loop $L^{(k)}$ and captures its computation semantics.

\stitle{Establishment of Single Static Assignment (SSA) -based multilevel graph.}
For each loop level $L^{(k)}$, let $\{w^{(k)}_0, w^{(k)}_1, \dots, w^{(k)}_{\tau_k}\}$ denote the set of variables used in its body. These include input arrays, intermediate temporaries, loop indices, and the output array. Without loss of generality, we assume the stencil kernel updates a single output array. If multiple outputs exist, they can be unified into a single abstract variable. Under this assumption, the output of $L^{(k-1)}$ naturally is a variable in $L^{(k)}$. We denote it as $w^{(k)}_0$, and conceptually treat $L^{(k-1)}$ as a functional computation of $w^{(k)}_0$ based on the inputs $\{w^{(k)}_i\}_{i=1}^{\tau_k}$. The multilevel data dependency graph is constructed through the following three steps:

\etitle{Step 1: Create vertices for individual variables.} 
Each variable $w^{(k)}_i$ may take different values at different iterations of $L^{(k)}$. Let $\vec{t}^{(k)}$ denote the iteration context of level $k$, including the iteration index $t_k$ and the outer loop context $(t_{k+1}, \dots, t_K)$. Based on the single static assignment (SSA) form, we denote $w^{(k,\vec{t}^{(k)})}_i$ as a variable instance of $w^{(k)}_i$ and represent it as a graph vertex $v^{(k,\vec{t}^{(k)})}_i$.

\etitle{Step 2: Construct data dependency edges.}  
From $\text{Statement}(Body~of~L^{(k)})$, we extract the computation pattern of each $w^{(k,\vec{t}^{(k)})}_j$ and its data dependencies. If $w^{(k,\vec{t}^{(k)})}_j$ depends on $w^{(k,\vec{t}'^{(k)})}_i$ (including spatial or temporal shifts), we add a directed edge from $v^{(k,\vec{t}'^{(k)})}_i$ to $v^{(k,\vec{t}^{(k)})}_j$. The resulting graph captures both intra- and inter-iteration data dependencies.

\etitle{Step 3: Assign attributes to vertices.}
Each vertex $v^{(k,\vec{t}^{(k)})}_i$ is annotated with the following attributes, i.e.,
(1) \textit{Data attribute}: the variable instance $w^{(k,\vec{t}^{(k)})}_i$;  
(2) \textit{Operation attribute}: computation operation that produces $w^{(k,\vec{t}^{(k)})}_i$;  
(3) \textit{Level attribute}: index $k$ indicating the depth of $L^{(k)}$;  
(4) \textit{Loop attribute}: an index vector that is a reduced form of $\vec{t}^{(k)}$ including two components: the call index representing how many times $L^{(k)}$ has been invoked by its outer loops, and the iteration index within $L^{(k)}$.

\stitle{SSA-based Multilevel Graph.} The formalized description of the established graph is as below.

\begin{definition}[SSA-based Multilevel Graph (SMG)]
For a stencil kernel represented by a $K$-level nested loop $L$, the associated  \emph{SSA-based multilevel graph} is a directed graph $G (V, E)$ organized into $K$ levels, each corresponding to a distinct loop level in  $L$. $V$ and $E$ represent the vertex and edge sets, respectively. Each vertex $v \in V$ represents a unique SSA-form definition or use of a data value occurring at a specific loop iteration. Each directed edge $(v, v') \in E$ denotes a data dependency, indicating that the data at $v'$ is computed using the data at $v$. Each $v \in V$ owns four attributes:
\\
\textbf{(1) Data attribute} $D(v) = (D_{\mathsf{type}}(v), D_{\mathsf{variable}}(v))$: Let $D: V \to \mathcal{T} \times \mathcal{V}$ be a function that assigns to each vertex $v$ its associated data information. Here, $\mathcal{T}$ denotes the set of data types, including $\mathrm{integer}$, $\mathrm{float}$, and $\mathrm{array}$, and $\mathcal{V}$ denotes the set of variable instances. 
\\
\textbf{(2) Operation attribute} $O(v)$:  Let $\mathrm{Pred}(v) = \{\, w \in V \mid (w, v) \in E \,\}$ denote the set of immediate predecessors of $v\in V$. The map $O$ then assigns to $v$ an operation function $O(v):=f$ such that $D_{\text{variable}}(v) = f\big(D_{\text{variable}}(w_1), \cdots, D_{\text{variable}}(w_n)\big)$, where $w_1, \cdots, w_n \in \mathrm{Pred}(v)$. For stencil kernels, all operation selections are concluded in Appendix \ref{section: Intermediate Representation}. 
\\
\textbf{ (3) Level attribute} $\ell(v)$:  
Let $\mathcal{K} = \{1, 2, \dots, K\}$. The map $\ell: V \to \mathcal{K}$ defines the level number $\ell(v) \in \mathcal{K}$ of $v$ that appears in the graph $G(V,E)$, corresponding to the depth of the loop where the corresponding variable instance is defined or used.
\\
\textbf{(4) Loop attribute} $\mathcal{L}(v) = (C(v),I(v))$ consists of two functions, $C: V \to \mathbb{N}$ and $I: V \to \mathbb{N}$, where $C(v)$ denotes the number of times $L^{(\ell(v))}$ has been invoked by its outer loops, and $I(v)$ records the intra-loop iteration index of $v$ within the current execution of $L^{(\ell(v))}$. 

\end{definition}

\etitle{Remark.} (1) Data dependencies across loop levels are represented in the SMG using special vertices $v$ annotated with $O(v) = \textit{Loopcall}$, indicating that $O(v)$ corresponds to the complete execution of the inner loop $L^{(\ell(v)-1)}$. These vertices abstract the invocation of inner loop bodies in the context of outer loops and serve as explicit interfaces between distinct nesting levels. (2) Boundary conditions that adjust iteration indices are modeled via $\varphi$-functions in SSA form. Specifically, for certain vertices $v'$, the operation $O(v')$ is set as $\varphi$ function, which selects the appropriate incoming value based on control flow or index guards at loop boundaries.

\begin{definition}[Level-$k$ SMG Subgraph]
Given an SMG $G(V, E)$, the \emph{Level-$k$ SMG subgraph}, denoted $G_{k}(V_{k}, E_{k})$, is a subgraph of $G(V,E)$ induced by the vertex set $V_{k} = \{v \in V \mid \ell(v) = k\}$,  where $E_{k} = \{ (u, v) \in E \mid u, v \in V_{k} \}$.
\end{definition}

Figure~\ref{fig:motivation-2}~(a) illustrates an example of the SMG constructed from the two-level nested stencil kernel shown in (d), including the corresponding level-1 and level-2 SMG subgraphs.

\subsubsection{ Invariant Graph and Summary Attribute}
\label{section: Invariant Graph and Summary Attribute}

In the following, we exploit the invariant structures in each $G_{k}(V_k,E_k)$ based on the common features of stencil computation.

\stitle{Regular loop.} Stencil kernels, whether from legacy codes or recent applications, exhibit structurally repetitive computation across loop iterations, with only minimal use of conditional branches within loop bodies. For example, even when boundary condition handling or tiling optimizations are applied, the number of conditional branches remains negligible compared to the total number of iterations during a complete loop execution. This indicates that the vast majority of iterations execute the same sequence of statements, differing only in iteration indices, a characteristic referred to as loop regularity in stencil computations. In the $K$-level nested loop $L$, $L^{(k)}$ is regular if most of the iterations in $L^{(k)}$ execute the same statements, except for a small number of iterations that follow exceptional conditional branches. This definition is formalized as follows:

\begin{definition}[Regular Loop]
\label{definition: regular loop}
 Let $M_k$ be the total number of calls of the $k$-th level loop $L^{(k)}$ invoked by its enclosing loops, and $N_k$ be the total number of iterations in the $k$-th level loop $L^{(k)}$. Denote $ST^{(m,n)}_k$ as the set of statements executed during the $n$-th iteration of $L^{(k)}$ at its $m$-th invocation. $L^{(k)}$ is regular if and only if there exists a set $ST^*_k \subset \text{Statement}(Body~of~L^{(k)})$ such that 
\\
(1) $r_k \ll M_k$ and $r_k \ll N_k$ with $r_k = \left| \left\{ (m,n) \mid ST^{(m,n)}_k \neq ST^{*}_k \text{ for } m \in [1,M_k] \text{ and } n \in [1,N_k] \right\} \right|$, i.e., the number of atypical iterations is negligible;
\\
(2) $\mu_k$ is a small positive integer, or $\mu_k \ll \chi_k$ with $\mu_k = \left| ST^{(m,n)}_k \setminus ST^{*}_k \right|$ \text{ and } $\chi_k = \left| ST^{*}_k\right|$, i.e., for each atypical iteration, the number of deviating statements is small.
\end{definition}

For example, in common stencil kernels, each loop level typically corresponds to updates along a single spatial dimension of the output array. In such cases, the set of executed statements $ST^{(m,n)}_k$ remains identical for most iterations, except possibly for $ST^{(m,1)}_k$ and $ST^{(m,N_k)}_k$, which may include additional statements to apply boundary conditions. These exceptional cases correspond to updates at domain boundaries, while remaining iterations perform uniform updates of interior mesh points.

\stitle{Construction of invariant subgraphs.} For each $ST^{(m,n)}_k$, we extract its associated subgraph in $G_k(V_k,E_k)$, representing the data dependencies induced by the statements executed in this iteration. The formal definition is given below.

\begin{definition}[Loop Slice]
In the Level-$k$ SMG subgraph $G_{k}(V_{k}, E_{k})$, a \emph{loop slice with Call-$m$ Iteration-$n$} is a subgraph $G_{k}^{(m,n)} (V_{k}^{(m,n)}, E_{k}^{(m,n)})$ induced by $V_{k}^{(m,n)} = \{v \in V_{k} \mid C(v)=m,~ I(v) = n\}$, where $E_{k}^{(m,n)} = \{(u, v) \in E_{k} \mid u, v \in V_{k}^{(m,n)}\}$.
\end{definition}

Based on the following two steps, we construct a data dependency graph that could represent the computation semantics for any $ST^{(m,n)}_{k}$. 

\etitle{Step 1: Establish a data dependency graph $G^*_k$ for $ST^{*}_{k}$.} With a similar process for the establishment of SMG, we construct a data dependency graph $G^*_k$ corresponding to $ST^{*}_{k}$. Most of $ST^{(m,n)}_{k}$ are equal to $ST^{*}_{k}$, indicating that most of $G^{(m,n)}_{k}$ share the same structure of $G^*_k$.

\etitle{Step 2: Expend the definition of $\varphi$ functions in $G^*_k$.} As $G^{*}_{k}$ captures the repeated computation structures of the level-$k$ loop $L^{(k)}$, it contains a vertex $v^*$ whose operation attribute $O(v^*)$ is a $\varphi$ function representing the loop condition of $L^{(k)}$. If $L^{(k)}$ is regular, Definition \ref{definition: regular loop} implies that only a small fraction of the $(m,n)$ iterations deviate from the common statements set $ST^*_k$. We thus extend the $\varphi$ function $O(v^*)$ to merge the divergent control paths arising from these exceptional iterations. With this extension, $G^*_k$ can uniformly represent the computation patterns of all $ST^{(m,n)}_k$ instances, including those involving conditional branches for boundary or special-case handling.

\stitle{Invariant subgraphs and invariant graph.} The preceding discussion establishes the existence of invariant subgraphs for regular loop $L^{(k)}$, which we now define formally.

\begin{definition}[Invariant Subgraph]
Given the level-$k$ SMG subgraph $G_k = (V_k, E_k)$ and its loop slices $G_k^{(m,n)} = (V_k^{(m,n)}, E_k^{(m,n)})$, an \emph{invariant subgraph} $G_k^* = (V_k^*, E_k^*)$ is a common structural pattern shared across all $G_k^{(m,n)}$, satisfying the following conditions:
\\
(1) For each $ST^{(m,n)}_k = ST^{*}_k$, there exists a bijective mapping $f^{(m,n)}_k: V_k^{(m,n)} \rightarrow V_k^*$ such that for any edge $(u, v) \in E_k^{(m,n)}$, we have $(f^{(m,n)}_k(u), f^{(m,n)}_k(v)) \in E_k^*$.
\\
(2) For any $i \neq j$, if an edge $(u, v) \in E_k$ with $u \in V_k^{(m,i)}$ and $v \in V_k^{(m,j)}$, $(f^{(m,i)}_k(u), f^{(m,j)}_k(v)) \in E_k^*$.
\\
(3) For the  $v^* \in V^{*}_{k}$ whose operation attribute involves the loop condition of $L^{(k)}$, $O(v^*)$ is a extended $\varphi$ function to merge branching execution paths for a small number of additional statements in $ST^{(m,n)}_k \setminus ST^{*}_k$ for $ST^{(m,n)}_k \neq ST^{*}_k$.
\end{definition}
By collecting invariant subgraphs from all levels and restoring cross-level edges, we construct the global \emph{Invariant Graph}, compactly encoding the hierarchical computation.

\begin{definition}[Invariant Graph]
Given a $K$-level SMG $G (V, E)$, let $G_k^* (V_k^*, E_k^*)$ denote the invariant subgraph at level-$k$, constructed from loop slices $\{G_k^{(m,n)}\}$ via bijective mappings $\{f^{(m,n)}_k\}$. The \emph{invariant graph} $G^*(V^*, E^*)$ is defined as: $V^* = \bigcup_{k=1}^{K} V_k^*$, and $E^*$ includes (1) all intra-level edges from each invariant subgraphs, i.e., $\bigcup_{k=1}^{K} E_k^*$; (2) all inter-level edges $(f_k^{(m,n)}(u), f_{k'}^{(m',n')}(v))$ where $u \in V_k^{(m,n)}$ and $v \in V_{k'}^{(m',n')}$ with $k \ne k'$.
\end{definition}

\etitle{Remark.} The invariant subgraph $G^*_k(V^*_k, E^*_k)$ captures the intra-level structure of  $G(V,E)$ at loop level $k$ by projecting all loop slices along the iteration direction, as illustrated in Figure \ref{fig:motivation-2}~(b). While cross-level dependencies are excluded from $G^*_k(V^*_k, E^*_k)$, they reconstructed globally in the invariant graph $G^*(V^*, E^*)$. Notably, $G^*_k(V^*_k, E^*_k)$ is not a directed acyclic graph (DAG), as the merging of ssa-based slices may introduce cycles when multiple data instances are mapped to the same vertex. The invariant subgraph captures the invariant process during the iterations at a given loop level, providing the analysis basis for Stencil-Lifting.  Additionally, we designe a specialized intermediate representation (IR) to express the entire invariant graph, as shown in Figure \ref{fig:motivation-2}~(c). The details of IR are presented in Appendix \ref{section: Intermediate Representation}.

\stitle{Summary attribute.} 
 We now specify the attributes associated with each vertex in the invariant subgraph. Each invariant subgraph vertex inherits four attributes from the SMG vertex: \emph{Data}, \emph{Operation}, \emph{Level} and \emph{Loop}. To support summary lifting, an additional attribute is introduced: \emph{Summary attributes} $S(v)$, represented as a logical predicate that summarizes the computational semantics for deriving $D_{\text{variable}}(v)$.

Because an invariant subgraph vertex abstracts multiple SMG vertices across loop iterations, the \emph{Data} and \emph{Loop} attributes of vertices in invariant subgraphs may vary during loop execution and are therefore treated as parameters. In contrast, the \emph{Operation} and \emph{Level} attributes remain constant, reflecting the invariant computational structure.

\subsection{Summary Syntax and Rule}

\subsubsection{ Summary Syntax}
\label{section: Summary Syntax}

\begin{table}
    \centering
\tiny
\begin{tabular}{llll}
\toprule[1.2pt]
$arrayInput$    & $::=\;$ array-type free variables created by input operation &
$numArithm$     & $::=\; num \;numOp\; num \;|\; numOp \; num$ \\
$indexInput$    & $::=\;$ index-type free variables created by input operation &
$indexCompare$  & $::=\; index \; (\; = \;|\; \neq \;|\; > \;|\; < \;|\; \le \;|\; \ge \;)\; index$ \\
$numInput$      & $::=\;$ numerical-type free variables created by input operation &
$indexArithm$   & $::=\; index \; (\, + \,|\,-\,|\,*\,|\,/\,|\,\%\,) \; index \; | \; - index$ \\
$condInput$     & $::=\;$ bool-type free variables created by input operation &
$boolArithm$    & $::=\; cond \; (\wedge|\vee) \; cond \;|\; \neg cond$ \\
$numOp$         & $::=\;$ numerical operators supported by SMT solvers &
$elemExp$       & $::=\;d num \;|\; index$\\
$num$           & $::=\; numInput \;|\; numArithm \;|\; arrayGet$ &
$cond$          & $::=\; indexCompare \;|\; boolArithm \;|\; condInput$ \\
$arrayGet$      & $::=\; \textbf{get}(arrayInput\;\{,\;index\})$ &
$summary$       & $::=\; \bigvee_i branch_i \;|\; elemExp$ \\
$index$         & $::=\; indexInput \;|\; indexArith$ &
$branch$        & $::=\;  cond \rightarrow elemExp$ \\

\bottomrule[1.2pt]
\end{tabular}
    \caption{Summary syntax in Stencil-Lifting.}
    \label{tab:summary-syntax}
    \vspace{-1.cm}
\end{table}

The vertex summary attribute consists of multiple branches, each defined by an element expression and a validity condition (denoted as $elemExp$ and $cond$ in Table~\ref{tab:summary-syntax}), both parameterized by loop indices. The element expression encodes the symbolic semantics for deriving the variable on each vertex, while the validity condition constrains its applicability to specific loop contexts, often distinguishing interior from boundary regions. All conditions must be mutually exclusive to ensure deterministic interpretation. The summary syntax details are shown in Table \ref{tab:summary-syntax}. The predicate syntax, tailored for stencil kernels, is designed to support: \emph{(1) Symbolic array representation}: Arrays are treated as symbolic entities to enable precise tracking of inter-element data dependencies. \emph{(2) Conditional semantics over loop regions}: Branching allows different symbolic expressions to describe the computation in distinct loop regions, particularly for handling boundary conditions. \emph{(3) Language-level translatability}: The summary syntax is compatible with domain-specific languages such as Halide, facilitating automated and efficient code generation.

\subsubsection{ Summary Rule}
\label{subsubsection:SummaryRule}

To ensure that $S(v)$ accurately captures the computational semantics of $D_{\text{variable}}(v)$, vertex summaries are generated based on predefined \emph{Summary Rules}. A summary rule derives $S(v)$ based on $O(v)$ and the summaries of vertices in $\text{Pred}(v)$. In Stencil-Lifting, we abstract seven summary rules covering common stencil operations (Figure \ref{fig:op-rule}).

\stitle{\emph{Rule 1}}: If \( O(v) \) is an $input$ operation, \( S(v) \)  is set to \( D_{\text{variable}}(v) \). This rule initializes the summaries of input vertices with their data.

\stitle{\emph{Rule 2}}: If \( O(v) \) is a $bound$ operation, \( S(v) \) is set as an index-typed \( D_{\text{variable}}(v) \), constrained to a range from \( begin \) to \( end \) with a step size of \( step \) for each iteration.

\stitle{\emph{Rule 3}}: If $v$ has two predecessors $v_i$ and $v_j$, where \( D_{\text{type}}(v_i) \) and \( D_{\text{type}}(v_j) \) are scalars (numerical, index, or boolean) and \( O(v) \) is a scalar operation (e.g., numerical, index, or boolean arithmetic), then $S(v) = S(v_i)\ O(v)\ S(v_j)$ (e.g., $S(v) = S(v_i) + S(v_j)$).

\stitle{\emph{Rule 4}}: If \( D_{\text{type}}(v_i) \) is array and \( D_{\text{type}}(v_j) \) is scalar, and \( O(v) \) is a $get$ operation (retrieving an array element by index), then $S(v) = S(v_i)[S(v_j)]$. For multi-dimensional arrays, \textit{Rule 4} involves additional predecessors of \( v \).  

\stitle{\emph{Rule 5}}: For a $set$ operation, \( S(v) \) represents an array state where only the element at index \( S(v_i) \) is updated by a numerical expression \( S(v_j) \), and other elements remains unchanged.

\stitle{\emph{Rule 6}}: Applies to \( \varphi \) functions, which aggregate variable definitions across different branches. \( S(v) \) merges the summaries of its predecessors as new branches.  

\stitle{\emph{Rule 7}}: Handles the integration of an inner loop's postcondition into the outer loop's invariant subgraph. As shown in Figure \ref{fig:op-rule}, if $v \in G^*_k(V^*_k,E^*_k)$ and \( O(v) \) is $Loopcall$, its predecessors (e.g., $v_a, v_b, \cdots$) are actual arguments passed into the inner loop. Each such predecessor has a corresponding \textit{dual vertex} (e.g., $v_{a'}, v_{b'}, \cdots$) in the inner loop's invariant subgraph $G^*_{k-1}(V^*_{k-1},E^*_{k-1})$, representing formal parameters. \( S(v) \) is set as the postcondition of the inner loop where each dual vertex's summaries is substituted with that of its corresponding outer loop vertex. 

\begin{figure}
    \centering
    \includegraphics[width=0.95\linewidth]{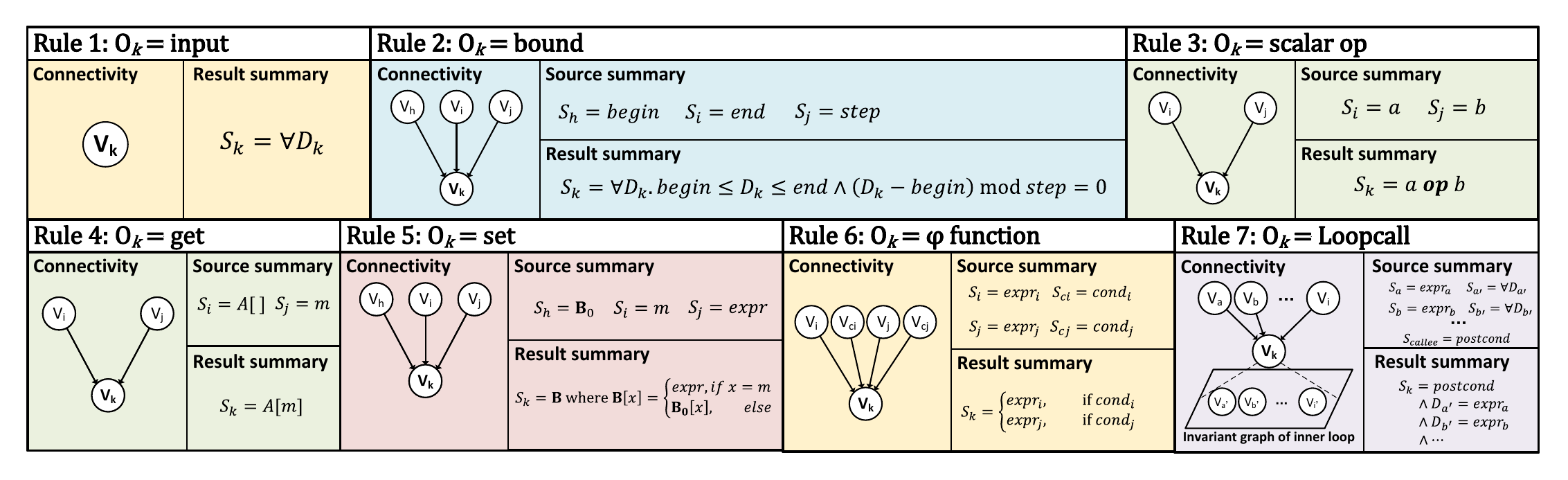}
    \vspace{-0.5cm}
    \caption{Summary rules in Stencil-Lifting.}
    \label{fig:op-rule}
    \vspace{-0.3cm}
\end{figure}

\subsection{Hierarchy Recursive Lifting Theory}
\label{subsec:HierarchyRecursiveLiftingTheory}
Based on the summary rules, we analyze the relationships among summaries of all vertices in the level-$k$ invariant subgraph $G^*_k(V^*_k,E^*_k)$. Let $V^*_k=\{v^k_1, v^k_2, \cdots, v^k_{h_k}\}$ denotes the set of $h_k$ vertices. A \emph{Summary Configuration} is defined as the tuple $(S(v^k_1), S(v^k_2), \cdots, S(v^k_{h_k}))$, where each $S(v^k_{i})$ is governed by a summary rule $R(v^k_i)$ associated with its operation $O(v^k_{i})$. Each rule $R(v^k_i)$ is one of the seven rules defined in Section~\ref{subsubsection:SummaryRule}. The rules $(R(v^k_1), R(v^k_2), \cdots, R(v^k_{h_k}))$ collectively impose a system of constraints over the summaries, determining a valid summary configuration.

\begin{definition}[Self-consistent Summary Configuration]
\label{definition: self-consistent configuration}
On the invariant subgraph $G^*_k(V^*_k,E^*_k)$, a summary configuration $(S^*(v^k_1), S^*(v^k_2), \cdots, S^*(v^k_{h_k}))$
is \textit{self-consistent} under a set of conjunctive rules $\{ R(v^k_1), R(v^k_2), \cdots, R(v^k_{h_k})\}$ if and only if, for every vertex $v^k_i \in V^*_k$ with its all predecessors $\{v^k_{i_1}, v^k_{i_2}, \cdots, v^k_{i_m}\}$, 
$S^*(v^k_i) = R(v^k_i) [S^*(v^k_{i_1}), S^*(v^k_{i_2}), \cdots, S^*(v^k_{i_m})]$. Each rule $R(v^k_i)$ is interpreted as a function that computes $S^*(v^k_i)$ based on the summaries of its predecessors.
\end{definition}

\etitle{Remark.} A \textit{self-consistent} summary configuration $(S^*(v^k_1), \cdots, S^*(v^k_{h_k}))$ provides a valid solution satisfying the constraints imposed by $R(v^k_1), R(v^k_2), \cdots, R(v^k_{h_k})$.

Let $v^k_{p_k}$ be the vertex whose operation $O(v^k_{p_k})$ is a $\varphi$ function with one input being the output array. Denote by $S^*(v^k_{p_k})|_{\mathcal{L}(v^k_{p_k})=(m,n)}$ the instance of $S^*(v^k_{p_k})$ for calculating $D_{variable}(v^k_{p_k})$ in the $n$-th iteration of loop $L^{(k)}$ at its $m$-th invocation. 

\begin{theorem}
\label{theorem: self-consistent}
Let $(S^*(v^k_1), S^*(v^k_2), \cdots, S^*(v^k_{h_k}))$ be a \textit{self-consistent} summary configuration under the $h_k$ rules $R(v^k_1), R(v^k_2), \cdots, R(v^k_{h_k})$. If $S^*(v^k_{p_k})|_{\mathcal{L}(v^k_{p_k})=(m,n)}$ is set as the loop invariant and $S^*(v^k_{p_k})|_{\mathcal{L}(v^k_{p_k})=(m,N_k)}$ is set as the postcondition for $L^{(k)}$ at its $m$-th invocation (where $1\leq m \leq M_k$ and $ 1 \leq n \leq N_k$), then the three statements in Equation (\ref{equation: VC}) hold.
\end{theorem}

\etitle{Remark.} Theorem \ref{theorem: self-consistent}, proved in Appendix~\ref{app:B}, shows that a \textit{self-consistent} summary configuration on $G^*_k(V^*_k,E^*_k)$ yiels the loop invariant and postcondition of loop $L^{(k)}$.

For any level $k \geq 2$, let $v^k_{q_k} \in V^*_k$ be the vertex with $O(v^k_{q_k}) = \mathit{Loopcall}$, representing a complete execution of loop $L^{(k-1)}$. The following theorem lays the theoretical foundation of Stencil-Lifting, motivating a hierarchical recursive lifting over $G^*(V^*,E^*)$.

\begin{theorem}
\label{theorem: recursive lifting}
For any invariant subgraph \( G^*_k(V^*_k,E^*_k) \) of \( G^*(V^*,E^*) \), let the summary configuration $(S^*(v^k_1), \cdots, S^*(v^k_{h_k}))$ be \textit{self-consistent} under rules $R(v^k_1), \cdots, R(v^k_{h_k})$. Suppose that
$v^k_{q_k}$ and $v^{k-1}_{p_{k-1}}$ satisfy $S^*(v^k_{q_k})|_{\mathcal{L}(v^k_{q_k})=(m_k,n_k)} = S^*(v^{k-1}_{p_{k-1}}) |_{\mathcal{L}(v^{k-1}_{p_{k-1}})=( ((m_k-1) \cdot N_{k} + n_{k}) \cdot N_{k-1},N_{k-1})}$ for any $m_k$ and $n_k$ ($1 \leq m_k \leq M_k$ and $1 \leq n_k \leq N_k$). $S^*(v^K_{p_K})|_{\mathcal{L}(v^K_{p_K})=(M_K,N_K)}$ represents the summary/postcondition of the stencil kernel.
\end{theorem}

\etitle{Remark.} Theorem \ref{theorem: recursive lifting}, proved in Appendix~\ref{app:B}, establishes that if (i) the summary configuration at each level is \textit{self-consistent}, and (ii) each vertex with $\mathit{Loopcall}$ operation at level $k$ uses the postcondition of the level-$(k-1)$ loop as its summary, then the summary of the top-level vertex associated with the $\varphi$ function yields the final postcondition of the entire stencil computation.

\section{Hierarchical Recursive Lifting Algorithm for Stencil Summary}
\label{section: Summary Lifting}

\subsection{Basic Idea}
\label{section: Basic Idea}

\begin{figure}
    \centering
    \includegraphics[width=0.89\linewidth]{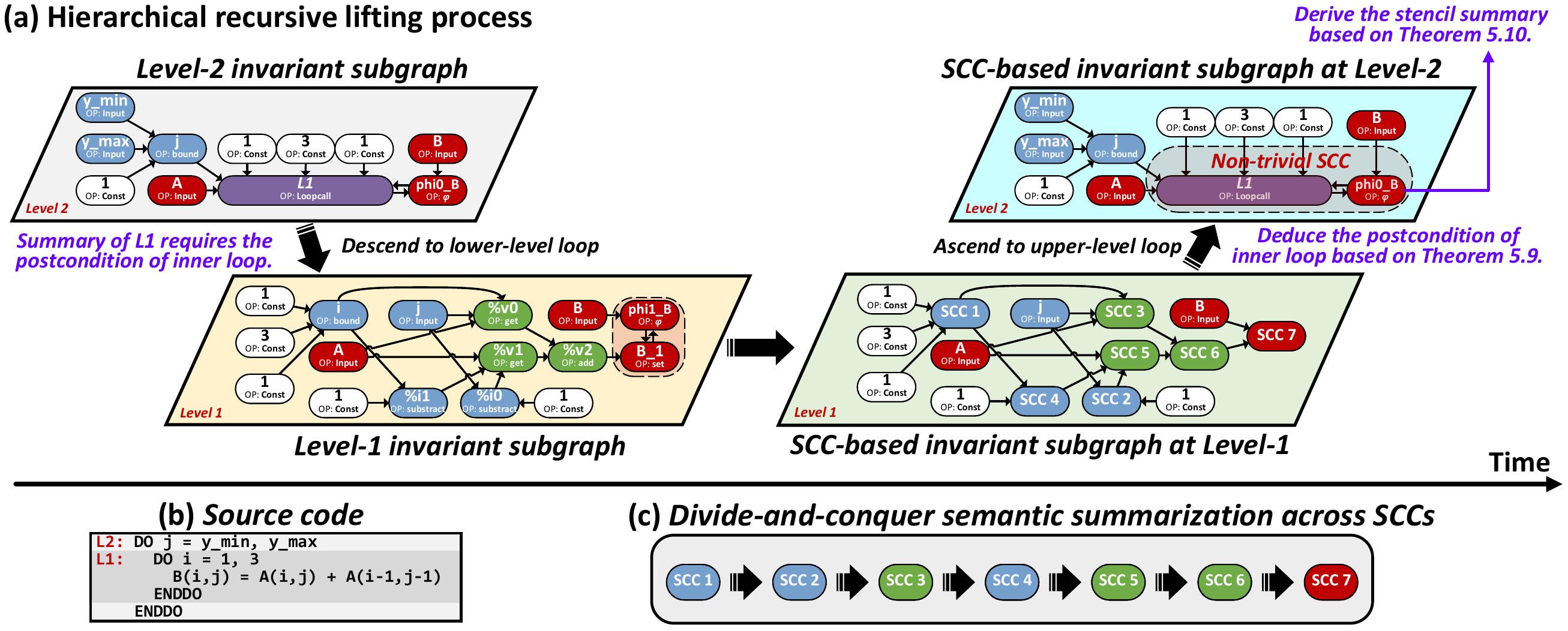}
    \vspace{-0.2em}
    \caption{Basic idea of our design.}
    \label{fig: Basic idea of our design}
    \vspace{-0.2cm}
\end{figure}

\subsubsection{ Hierarchical Recursive Lifting}
Based on Theorem \ref{theorem: recursive lifting}, the key idea of our algorithm design is to compute the stencil summary by recursively performing semantic extractions across all levels of the invariant graph $G^*(V^*,E^*)$. Specifically, at the level-$k$ invariant subgraph $G^*_k(V^*_k,E^*_k)$ associated with loop $L^{(k)}$, a divide-and-conquer semantic summarization is designed to derive a \textit{self-consistent} summary configuration. When a vertex annotated with a $Loopcall$ operation is encountered, the summarization process recursively descends into the lower-level invariant subgraph $G^*_{k-1}(V^*_{k-1}, E^*_{k-1})$ to compute \emph{self-consistent} summaries and the postcondition of $L^{(k-1)}$. 

Once the semantic summarization over $G^*_k(V^*_k, E^*_k)$ completes, Theorem \ref{theorem: self-consistent} guarantees the \textit{self-consistent} summary of the output-related $\varphi$-vertex $v^{k}_{p_{k}}$ (Section~\ref{subsec:HierarchyRecursiveLiftingTheory}) represents the postcondition of the leve-$k$ loop $L^{(k)}$. This postcondition is then propagated to the level $k+1$ by applying the summary rule of the $Loopcall$ vertex $v^{k+1}_{q_{k+1}}$ (Section~\ref{subsec:HierarchyRecursiveLiftingTheory}) in $G^*_{k+1}(V^*_{k+1}, E^*_{k+1})$. In this way, the lifting proceeds recursively from lower to higher levels, with the postcondition of the top-level loop ultimately forming the overall summary of the stencil kernel (Theorem \ref{theorem: recursive lifting}). An example of the hierarchical recursive lifting process for a two-level nested loop is illustrated in Figure~\ref{fig: Basic idea of our design}(a), derived from the source code presented in Figure~\ref{fig: Basic idea of our design}(b).

\subsubsection{ Divide-and-Conquer Semantic Summarization}
To derive the postcondition of the $k$-th level loop in the target stencil kernel, we design a divide-and-conquer semantic summarization.

\stitle{Fundamentals.}
As discussed in Section \ref{section: Lifting Theory}, $G^*_k(V^*_k,E^*_k)$ may contain multiple cycle structures, leading to complex inter-dependencies among vertices and complicating semantic summarization. To facilitate a structured analysis, we partition $G^*_k$ into its strongly connected components (SCCs):

\begin{definition}[Strongly Connected Component (SCC)]
\label{def:scc}
A strongly connected component of a directed graph is a maximal subgraph in which every vertex is reachable from every other vertex.
\end{definition}

According to Definition~\ref{def:scc}, each cycle in $G^*_k(V^*_k, E^*_k)$ is naturally encapsulated within an SCC. In a trivial case, a single vertex forms a singleton SCC, since it is reachable from itself. By treating SCCs as indivisible units, we obtain a higher-level abstraction of the original invariant subgraph.

\begin{definition}[SCC-based Invariant Subgraph]
\label{def:scc-basedig}
Given an invariant subgraph $G^*_k(V^*_k, E^*_k)$, its \emph{SCC-based invariant subgraph} is the reduced graph $\mathcal{G}^*_k = (\mathcal{V}^*_k, \mathcal{E}^*_k)$, where each vertex in $\mathcal{V}^*_k$ represents either an SCC of $G^*_k$ or a singleton vertex not part of any cycle. There is an edge $(\omega, \nu) \in \mathcal{E}^*_k$ if there exists an edge $(w, v) \in E^*_k$ such that $w \in \omega$ and $v \in \nu$, with $\omega \ne \nu$.

\end{definition}

According to \emph{Definition~\ref{def:scc-basedig}}, $\mathcal{G}^*_k$ is a directed acyclic graph (DAG) by construction and captures the high-level dependency structure among SCCs and acyclic vertices in $G^*_k$. Stencil-Lifting applies Kosaraju's algorithm \cite{Sharir1981} to identify all SCCs within each invariant subgraph, thereby constructing the SCC-based invariant subgraph. For $G^*_k(V^*_k, E^*_k)$, the computational complexity of Kosaraju's algorithm is $O(|V^*_{k}|)$.

\etitle{Remark.} Certain vertices at higher loop levels may serve as direct predecessors of some vertices in $G^*_k(V^*_k, E^*_k)$. These vertices are referred to as \emph{external input vertices}. Although they are not part of the SCC-based subgraph, they act as sources of initialized summaries of input vertices required for the semantic summarization within $G^*_k(V^*_k, E^*_k)$.

\stitle{Divide-and-conquer mechanism.}
In $G^*_k(V^*_k, E^*_k)$, the semantic summarization follows a divide-and-conquer mechanism: \textit{decomposing the construction of a self-consistent summary configuration into a sequence of localized semantic extraction procedures over individual SCCs}, as shown in Figure~\ref{fig: Basic idea of our design}(c). 

This mechanism is enabled due to two structural properties of the SCC-based invariant subgraph. \emph{First}, the DAG structure of inter-SCC dependencies allows us to impose a topological order over all SCCs, ensuring that each SCC depends only on its predecessors. During the traversal of SCCs in this order, the \textit{self-consistent} summaries established for vertices in earlier SCCs could be directly used to determine the \textit{self-consistent} summaries for vertices in later SCCs. \emph{Second}, the computation of \textit{self-consistent} summaries within any given SCC is entirely independent of SCCs that follow it in the topological order.

\begin{minipage}{0.53\textwidth}
    \centering
    \scriptsize
    \begin{algorithm}[H]
    \caption{Hierarchical Recursive Lifting}
    \label{alg:recursive-semantic-lifting}
    \begin{algorithmic}[1]
    \Require A tensor $S$ storing summaries and invariant subgraph $G^*_k(V^*_k,E^*_k)$;
\Ensure the postcondition of the level-$k$ loop $L^{(k)}$;
\Function{HierarchicalRecursiveLifting}{$S$, $G^*_k$}
    \State $(\mathcal{G}^*_k, \Omega_k) \leftarrow$ \Call{BuildSCCGraphandSequence}{$G^*_k$}; {\color{blue} \Comment{Phase 1.}}
    \State \textbf{initialize} $S(v) \gets D_{\text{variable}}(v)$ for every input vertex $v$ of $G^*_k$;
    \ForAll{$SCC \in \Omega_k$ \textbf{in priority order}}
        \If{$SCC$ has a vertex $v$ with $O(v)=Loopcall$} {\color{blue}\Comment{Phase 2.}}
            \ForAll{$v \in SCC$ \textbf{with} $O(v)=Loopcall$}
                \State $G^*_{k-1} \!\gets$ \Call{LowerLevelInvariSubgraph}{$v$};
                {\color{red}\Comment{Descend.}}
                \State $\textit{post}\gets$\Call{HierarchicalRecursiveLifting}{$S$, $G^*_{k-1}$};
                \State \textbf{set} $S(v)$ with $\textit{post}$ and summaries of $v$'s predecessors;
            \EndFor
        \EndIf
        \State $S \gets$ \Call{IterativeExtractionWithinSCC}{${\sf SCC}$, $S$};
            {\color{blue}\Comment{Phase 3.}}
    \EndFor
    \State \Return $S(v^{*})$ where $v^{*}$ has $\varphi$ operation with output array as data;
\EndFunction

    \end{algorithmic}
    \end{algorithm}
    \vspace{0.1em}
\end{minipage}
\hfill
\begin{minipage}{0.43\textwidth}
    \centering
    \scriptsize
    \begin{algorithm}[H]
    \caption{Iterative Semantic Extraction}
    \label{alg:semantic-extraction-scc}
    \begin{algorithmic}[1]
\Require current SCC $\mathit{SCC}$ and tensor $S$ storing summaries;
\Ensure $S$ with \emph{self-consistent} summary $S(v)$ for $v \in SCC$;
\Function{IterativeExtractionWithinSCC}{$SCC$, $S$}
\State {\color{red} \textsc{VertexElimination}($SCC$,$S$); \Comment{Optional.}}
\State \textbf{select} a start vertex $v_{\text{start}}$; {\color{blue}\Comment{Step 1.}}
\State $S(v_{\text{start}}) \gets$ \textsc{InitializeSummary}($v_{\text{start}}$);
\State \textsc{ForwardSweep}($\mathit{SCC}$, $S$, $v_{\text{start}}$);
    {\color{blue}\Comment{Step 2.}}
\Repeat
    \State {\color{red} \textsc{EquivalenceChecking}($S(v_{\text{start}})$); \Comment{Optional.}}
    \State $S^{\text{general}} \gets$ \textsc{Generalize}($S(v_{\text{start}})$); {\color{blue}\Comment{Step 3.}}
    \State $S^{(t-1)} \gets$ \textsc{Shift}($S^{\text{general}}$, $t \gets t - 1$); {\color{blue}\Comment{Step 4.}}
    \State $S(v_{\text{start}}) \gets S^{(t-1)}$;
    \State \textsc{ForwardSweep}($\mathit{SCC}$, $S$, $v_{\text{start}}$);
    \State $S^{(t)} \gets S(v_{\text{start}})$;
\Until{$S^{(t)} == S^{\text{general}}$} {\color{blue}\Comment{Step 5.}}
\State \Return $S$;
\EndFunction
    \end{algorithmic}
    \end{algorithm}
    \vspace{0.1em}
\end{minipage}

\vspace{-1.1em}
\subsection{Hierarchical Recursive Lifting Algorithm}
Based on the basic idea, we propose a \emph{hierarchical recursive lifting algorithm} (Algorithm~\ref{alg:recursive-semantic-lifting}) that performs divide-and-conquer semantic summarization over the invariant subgraph $G^*_k(V^*_k, E^*_k)$ to derive the postcondition summary of loop $L^{(k)}$. For a stencil kernel with $K$ levels of nested loop, the overall postcondition (i.e., summary) is obtained by applying the algorithm with the top-level invariant subgraph $G^*_K(V^*_K, E^*_K)$ as input. The hierarchical recursive lifting proceeds as follows:

\textbf{\emph{Phase 1:}} \emph{Construct the SCC-based invariant subgraph of $G^*_k(V^*_k, E^*_k)$.} To enable the divide-and-conquer semantic summarization in $G^*_k(V^*_k, E^*_k)$, the SCC-based invariant subgraph $\mathcal{G}^*_k (\mathcal{V}^*_k, \mathcal{E}^*_k)$ is constructed using Kosaraju's algorithm. All SCCs are sorted into a sequence $\Omega_k$, such that if $(\omega,\nu) \in \mathcal{E}^*_k$, then $\omega$ precedes $\nu$ in $\Omega_k$. We retrieve the first SCC from $\Omega_k$ and denote it as $SCC_{\text{current}}$.

\textbf{\emph{Phase 2:}} \emph{Initialize summaries for vertices annotated with $Loopcall$ operations in $SCC_{\text{current}}$.} We check whether $SCC_{\text{current}}$ contains any vertex $v$ such that $O(v) = Loopcall$. If so, the algorithm recursively invokes itself on the corresponding lower-level invariant subgraph $G^*_{k-1}(V^*_{k-1}, E^*_{k-1})$ to compute the postcondition of loop $L^{(k-1)}$. The result is assigned to $S(v)$ following the summary rule associated with $v$. If no such vertex exists, the summary lifting proceeds to \emph{Phase 3}.

\textbf{\emph{Phase 3:}} \emph{Perform an iterative semantic extraction within $SCC_{\text{current}}$.}  
After initializing summaries for vertices whose operations are annotated as $Loopcall$, we perform an iterative semantic extraction to compute self-consistent summaries for all vertices in $SCC_{\text{current}}$ (details in Section \ref{subsection:self-consisten-summary-extraction}). Once convergence is reached, we check whether $\Omega_k$ is empty. If not, the next SCC is retrieved from $\Omega_k$, assigned to $SCC_{\text{current}}$, and the process returns to \emph{Phase 2}. If $\Omega_k$ is empty, return the summary of $v^{k}_{p_k}$, which represents the postcondition of $L^{(k)}$.

We now present the correctness and convergence guarantee of the recursive lifting process:

\begin{theorem}
\label{theorem: hierarchical recursive lifting}
Given the invariant graph for a stencil, the hierarchical recursive lifting algorithm possesses the finite termination property and yields the summary/postcondition of the stencil kernel.
\end{theorem}

\etitle{Remark.} The proof of Theorem \ref{theorem: hierarchical recursive lifting} is provided in Appendix \ref{Asection: Proofs 2}. This result ensures that, once the hierarchical recursive lifting completes, the derived summary can be directly regarded as the postcondition of the stencil kernel, without requiring further verification.

\subsection{Iterative Semantic Extraction}
\label{subsection:self-consisten-summary-extraction}

To construct the \emph{self-consistent} summaries of all vertices within each SCC, which is the core of \emph{Phase 3} in Algorithm~\ref{alg:recursive-semantic-lifting}, we propose an iterative semantic extraction strategy (Algorithm~\ref{alg:semantic-extraction-scc}).

First, for a trivial SCC, which contains only a single vertex, the self-consistent summary of the vertex can be directly obtained by applying its associated summary rule. For instance, in Figure \ref{fig: Basic idea of our design}(c), the first six SCCs are trivial, and their summaries are computed directly in this way.

Second, for a non-trivial SCC, Algorithm~\ref{alg:semantic-extraction-scc} aims to derive \emph{self-consistent} summaries for all vertices by applying their respective summary rules. In the presence of cycles, this is achieved through an iterative refinement (called a sweep below) until a fixed point is reached across the entire SCC.
The iterative process consists of five steps (Figure \ref{fig:regular-scc}):

\textbf{\emph{Step 1: Select and initialize the start vertex for sweep.}}
To begin the sweep, select a vertex  $v_{\text{start}}$ in the current SCC $SCC_{\text{current}}$ that has at least one predecessor outside $SCC_{\text{current}}$ with a known summary $\hat{S}$ (By topology order, the summary $\hat{S}$ of any such external predecessor has already been computed prior to the current SCC's processing). Preferably, select a vertex whose operation is a $\varphi$ function, as such vertices are guaranteed to have external predecessors. If no $\varphi$-vertex is available, select any other vertex meeting the above condition and initialize its summary similarly. Then, initialize the summary of $v_{\text{start}}$ as $S(v_{\text{start}}) = \hat{S}$.

\textbf{\emph{Step 2: Execute a forward sweep.}} 
Starting from $v_{\text{start}}$, traverse the current SCC along the direction of edges. For each visited vertex $v$, its predecessors fall into two categories: (1) vertices in other SCCs whose summaries have already been computed (by topological order), and (2) vertices within the current SCC that have been visited earlier in the current sweep.

Since all predecessors' summaries are  available when $v$ is visited, the summary rule of $v$ can be directly applied to compute its summary $S(v)$. The sweep propagates summary information forward through the SCC, eventually returning to $v_{\text{start}}$, whose summary is updated last.  

\etitle{Remark.} \emph{(Branch structure in the summary of $v_{\text{start}}$)}
After the sweep, the summary $S(v_{\text{start}})$ typically comprises multiple expressions, each guarded by a distinct validity condition. As illustrated in Figure~\ref{fig:regular-scc}, these guarded expressions originate from both external predecessors outside $SCC_{\text{current}}$ and internal predecessors within $SCC_{\text{current}}$. The overall branching structure is organized by the semantics of the $\varphi$ function. However, each guarded expression usually describes the update of an individual element in the output array. To make the summary applicable over a contiguous region of the output array rather than a single element, the next step generalizes these expressions.

\textbf{\emph{Step 3: Generalize the summary of the start vertex.}} Given that loop indices are treated as symbolic variables in the invariant subgraph (Observation 1), we generalize the summary of the start vertex to describe updates over array regions rather than individual elements.

Assume that the current summary of $v_{\text{start}}$ consists of $N_e + 1$ conditional branches. The first $N_e$ branches are expression–condition pairs $\{ (E_k(\mathbf{x}, t), P_k(\mathbf{x}, t)) \}_{k = 1}^{N_e}$, where each expression $E_k(\mathbf{x}, t)$ specifies an update to the array element at position $\mathbf{x}$ of $D_{\text{variable}}(v_{\text{start}})$ during iteration $t = I(v_{\text{start}})$ and $P_k(\mathbf{x}, t)$ denotes the validity condition under which this update applies. The final branch encodes the initialized summary of $v_{\text{start}}$, applied when none of the previous $P_k(\mathbf{x}, t)$ conditions hold. Here, $\mathbf{x} \in \mathbb{Z}^d$ is a $d$-dimensional array index, i.e., $\mathbf{x} = (x_1, x_2, \dots, x_d)$, and $t$ denotes the loop iteration variable. 
To generalize from a point-wise to a region-wise description, we perform a dataflow analysis over the array and introduce a predicate logic equation:
\begin{equation}
\label{equation: predicate logic equation}
\exists t_1 (0 \leq t_1 \wedge t_1 \leq t \wedge P_k(\mathbf{x},t_1) ) \wedge \neg \exists t_2 ( t_1 \leq t_2 \wedge t_2 \leq t \wedge P(\mathbf{x}, t_2)),
\end{equation}
where $P(\mathbf{x},t_2) = \vee^{N_e}_{k=1} P_k(\mathbf{x},t_2)$. Equation (\ref{equation: predicate logic equation}) characterizes all the array positions $\mathbf{x}$ of $D_{\text{variable}}(v_{\text{start}})$ that are last updated by $E_k$ before some iteration $t_1 \leq t$ and not overwritten by any other expression after iteration $t_1$.

In stencil computations, at a given loop level, array indices often depend linearly on the iteration variable, i.e., $x_{i} = a_{i}t + b_{i}$. Substituting this relation into Equation (\ref{equation: predicate logic equation}) and applying Skolemization \cite{Wintersteiger2010}, we derive lower and upper bounds ($LB^k_{i}(t)$ and $UB^{k}_{i}(t)$) for each $x_i$ as a function of the iteration variable $t$, leading to $\tilde{P}_k(\mathbf{x},t) =  \wedge^{d}_{i=1} (LB^k_{i}(t) \leq x_i \leq UB^{k}_{i}(t))$.
Next, we eliminate $t$ from each expression $E_k(\mathbf{x},t)$ by substituting $t=(x_i-b_i)/a_i$, yielding $t$-free generalized expression $\tilde{E}_k(\mathbf{x})$. With the generalized expression–condition pairs $\{\tilde{E}_k(\mathbf{x}), \tilde{P}_k(\mathbf{x}, t)) \}_{k = 1}^{N_e}$, the summary of the start vertex is rewritten into the following form that $S^{\text{general}}(v_{\text{start}}):{D}_{\text{variable}}(v_{\text{start}})$ where the element at each position $\mathbf{x}$ of the array ${D}_{\text{variable}}(v_{\text{start}})$ is given by:
\begin{displaymath}
{D}_{\text{variable}}(v_{\text{start}})[\mathbf{x}] =
\begin{cases}
\tilde{E}_1(\mathbf{x}), ~\hbox{if}~LB^1_{1}(t) \leq x_1 \leq UB^{1}_{1}(t) \wedge \cdots \wedge LB^1_{d}(t) \leq x_d \leq UB^{1}_{d}(t); \\
\tilde{E}_2(\mathbf{x}), ~\hbox{if}~LB^2_{1}(t) \leq x_1 \leq UB^2_{1}(t) \wedge \cdots \wedge LB^2_{d}(t) \leq x_d \leq UB^2_{d}(t); \\
\cdots \cdots; \\
\tilde{E}_{N_e}(\mathbf{x}), ~\hbox{if}~LB^{N_e}_{1}(t) \leq x_1 \leq UB^{N_e}_{1}(t) \wedge \cdots \wedge LB^{N_e}_{d}(t) \leq x_d \leq UB^{N_e}_{d}(t); \\
\hat{S}(\mathbf{x}), ~\hbox{else}.
\end{cases}
\end{displaymath}
This generalized summary describes the update over array regions rather than individual elements.

\textbf{\emph{Step 4: Perform a further sweep using the generalized summary.}} To advance  a \textit{self-consistent} summary configuration, a new sweep is initiated from $v_{\text{start}}$ using the generalized summary $S^{\text{general}}(v_{\text{start}})$ as input. Specially, all occurrences of the iteration variable $t$ in $S^{\text{general}}(v_{\text{start}})$ are replaced by $t-1$, producing a time shifted version denoted as $S^{(t - 1)}(v_{\text{start}})$. After using $S^{(t - 1)}(v_{\text{start}})$ to initialize $v_{\text{start}}$, we execute a full sweep over the SCC, as described in \emph{Step~2}, and obtain an updated summary $S^{(t)}(v_{\text{start}})$ at the end of the traversal.

\textbf{\emph{Step 5: Check for convergence.}} If the updated summary $S^{(t)}(v_{\text{start}})$ matches the previously generalized form $S^{general}(v_{\text{start}})$, then a fixed point has been reached. This indicates that the summary of $v_{\text{start}}$ is stable under further sweeps and, by extension, the summaries of all vertices in the SCC are \textit{self-consistent} under their respective summary rules. The iterative semantic extraction for this SCC therefore terminates. Otherwise, $S^{(t)}(v_{\text{start}})$ is used as the input to \emph{Step 3} for the next round of generalization, and \emph{Step 4} is repeated for further summary refinement.
 
The following theorem ensures that the semantic extraction process always terminates after a finite number of iterations.

\begin{figure}
    \centering
    \includegraphics[width=0.94\linewidth]{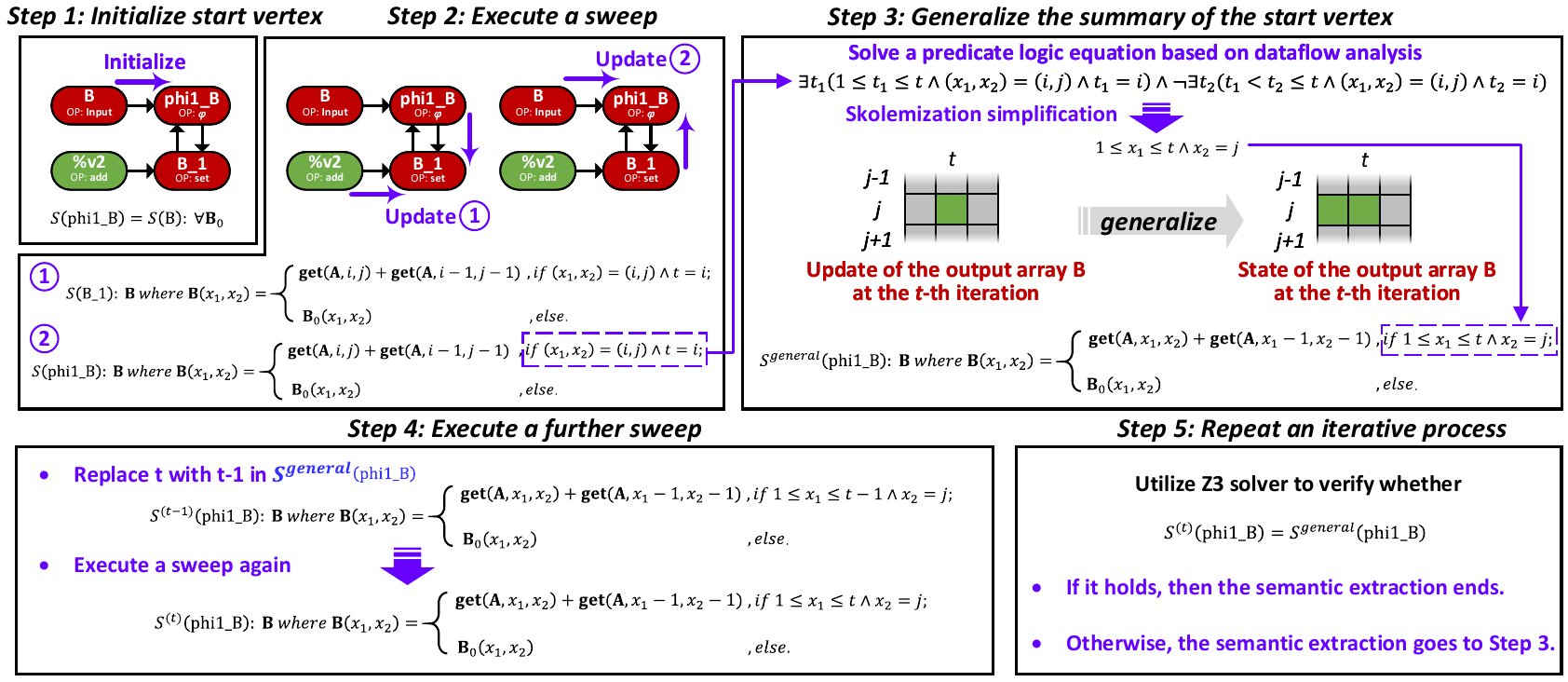}
    \vspace{-0.3cm}
    \caption{Iterative semantic extraction within SCC 7 of the SCC-based invariant graph shown in Figure \ref{fig: Basic idea of our design}. The detailed explanation about this example is in Appendix \ref{Asection: Hierarchical Recursive Lifting Algorithm}.}
    \label{fig:regular-scc}
    \vspace{-0.3cm}
\end{figure}

\etitle{Remark.} Stencil-Lifting utilizes the Z3 solver \cite{DeMoura2008} at \textit{Step 5} to verify whether $S^{(t)}(v_{\text{start}}) = S^{general}(v_{\text{start}})$ holds. 

\begin{theorem}[Finite Termination Property]
\label{theorem: finite termination property}
For any SCC in the invariant subgraph corresponding to a loop in a $d$-dimensional $m$-point stencil computation, there exists a positive integer $n_{iter}$ such that the iterative semantic extraction process converges after executing $n_{iter}$ sweeps.
\end{theorem}

\etitle{Remark.} According to the proof of Theorem \ref{theorem: finite termination property} in Appendix \ref{Asection: Proofs 2}, the finite termination property of the iterative semantic extraction process relies on a fundamental feature of stencils: the output array is updated based on a fixed number of elements from the input array, as specified by Equation~(\ref{equation: stencil computation pattern}).

\subsection{Details in Our Algorithm for Complex Cases}

This section presents how the hierarchical recursive lifting algorithm handles complex cases commonly found in legacy code. These include: (1) loop with conditional branches; (2) loop with parallelization directives; (3) loop with intermediate buffers; (4) unrolled loop; and (5) tiled loop.

Notably, our algorithm is inherently capable of handling all five cases without modification. To further improve the performance of the lifting process, we design modular acceleration techniques for handling unrolled and tiled loops, i.e., equivalence checking and vertex elimination, which are optional in Algorithm \ref{alg:semantic-extraction-scc} and integrate seamlessly into the recursive lifting process. 

\subsubsection{ Loop with Condition Branches}
\label{section: Loop with Condition Branches}
Conditional branches frequently appear in stencil codes to distinguish boundary and interior regions within fused loops. These branches introduce control-flow divergence, complicating program analysis and semantic extraction.

\stitle{Algorithmic support for conditional control flow.} The hierarchical recursive lifting algorithm addresses this challenge through two key designs mentioned in Section \ref{section: Lifting Theory}. 

\emph{Design 1: Condition-Aware Summary Syntax.} Each vertex in the invariant subgraph is annotated with a summary that explicitly includes conditionals (Section~\ref{section: Summary Syntax}). This enables accurate encoding of region-specific computation semantics as a function of loop indices.

\emph{Design 2: Extended $\varphi$ function for Divergent Paths.} To represent value merges along diverging control flows, Stencil-Lifting generalizes the SSA-style $\varphi$ functions (Section~\ref{section: Invariant Graph and Summary Attribute}). The extended $\varphi$ captures data values from conditional branches, enabling the invariant subgraph to faithfully represent control-dependent behavior.

\stitle{Example.} Figure~\ref{fig:case-study}(a) shows the construction of an extended $\varphi$ function for a nested loop with a conditional. The conditional expression $cond$ is initialized as \texttt{True} at loop entry and updated based on the control predicates. At the join point, merging all its values yields an extended $\varphi$ function.

\subsubsection{ Loop with Parallelization Directives}
Parallelization directives are frequently used in stencil kernels to exploit thread-level parallelism for high performance.

\stitle{Handling loops with parallelization directives.} In typical stencil programs, these directives enable independent execution of iterations without introducing inter-thread dependencies such as reductions or shared data conflicts. Thus, they do not alter the loop's dataflow structure captured by the invariant graph. Stencil-Lifting handles such cases by stripping parallelization directives during program analysis and constructing the invariant graph based solely on the loop’s original semantics. The hierarchical lifting algorithm then proceeds unchanged.

\stitle{Example.} As shown in Figure \ref{fig:case-study}(b), the parallelization directives could be removed before the construction of the invariant graph. Since no inter-thread dependency exists, the original computation semantics are preserved, enabling correct summary extraction.

\subsubsection{ Loop with Intermediate Buffers.} Stencil kernels often introduce intermediate buffers to support pipelined computation. While this improves performance, such buffers induce indirect data dependencies that complicate semantic extraction. 

\stitle{Resolving indirect data dependencies.} In \emph{Phase 3} of Algorithm \ref{alg:recursive-semantic-lifting}, the iterative semantic extraction within SCCs addresses this issue (Section \ref{subsection:self-consisten-summary-extraction}). By repeatedly sweeping through vertices in an SCC, the semantic extraction process progressively unfolds indirect dependencies. As vertex summaries converge, the final computation semantics emerge without requiring special treatment for buffers.

\stitle{Example.} Figure~\ref{fig:case-study}(c) shows how indirect dependencies introduced by buffers are resolved through iterative semantic extraction (Figure~\ref{fig:irreg-scc} in Appendix \ref{Asection: Hierarchical Recursive Lifting Algorithm} presents the detailed SCC-based invariant subgraph for the example). After the first sweep, the summary of $phi\_B$ depends on $B\_2$, which still contains an unresolved buffer-based dependency. In the second sweep, updated summaries propagate through the SCC, unfolding the indirect dependency. When the traversal revisits $phi\_B$, the dependency is resolved. A final sweep confirms convergence, yielding dependency-free summaries for all the vertices.

\subsubsection{ Unrolled Loop.}
\label{section: Unrolled Loops} 
Loop unrolling expands multiple iterations into explicit statements, improving instruction-level parallelism but increasing register pressure and structural redundancy.

\stitle{Handling redundancy introduced by loop unrolling.} Loop unrolling causes each invariant subgraph to contain multiple vertices that would have belonged to separate subgraphs in the original unrolled loop. This fact just introduces redundant branches in the vertex summaries and the final postcondition, while it does not interfere with the hierarchical recursive lifting process. Therefore, our algorithm remains applicable without any modification.

\stitle{Accelerating the semantic extraction via an equivalence checking.}
To reduce redundant branches in the vertex summaries and speed up the summary generalization step (\emph{Step 3}) in iterative semantic extraction, \emph{equivalence check} is applied before generalizing summaries in Algorithm~\ref{alg:semantic-extraction-scc}. Specifically, for each expression \( E_k \), we modify its position vector \( \mathbf{x} \) (replacing \( x_i \) with \( x_i \pm 1 \)) and check whether the resulting expression is equivalent to another expression \( E_{k'} \). If so, they are considered equivalent and merged. This reduces redundancy and improves convergence.

\stitle{Example.}
Figure~\ref{fig:case-study}(d) shows that in the summary of \( phi\_B \), replacing \( i \) with \( i+1 \) in \( E_1 \) yields an expression equivalent to \( E_2 \); likewise, changing $i$ in \( E_2 \) yields an expression equivalent to \( E_3 \). All three expressions are thus merged to eliminate redundancy in the summary.

\subsubsection{ Tiled Loop}
Loop tiling improves cache efficiency by restructuring a loop into nested inter-tile and intra-tile loops, enabling better spatial locality in memory access.

\stitle{Handling tiled loops.}
Tiling increases loop-nesting depth, adding extra levels to the invariant graph. These levels may update along the same spatial dimension. As the summary generalization (\emph{Step~3}) in semantic extraction merges semantically equivalent computation expressions across levels, the recursive lifting algorithm naturally handles tiled loops, requiring no special treatment.

\stitle{Vertex elimination to improve efficiency.}
To accelerate the summary lifting for tiled loops, we introduce vertex elimination. If the loop index at the current level matches the index used in the validity condition of a lower-level postcondition (i.e., both refer to the same output array dimension), we eliminate the vertices with $\varphi$ operations at the current level. Instead of extracting semantics within the SCC containing such vertices, we propagate the lower-level postcondition upward, adjusting the index accordingly.

\stitle{Example.} As shown in Figure \ref{fig:case-study}(e), when the algorithm reaches the Level-2 invariant subgraph, the loop index \( ti \) matches the index variable in the Level-1 postcondition, with both modifying array \( B \) along the same dimension. This indicates that Levels 1 and 2 represent the intra-tile and inter-tile loops of a tiled structure. To accelerate lifting, we bypass semantic extraction within the non-trivial SCC formed by $L1$ and $phi\_B$. Specifically, the vertex $phi\_B$ is eliminated, and vertex $B$ is directly linked to $L1$, forming a new invariant subgraph. According to the $Loopcall$ summary rule, the Level-1 postcondition is modified using the loop bounds and stride of \( ti \), which enables direct derivation of a self-consistent summary for $L1$.

\begin{figure}
    \centering
    \includegraphics[width=0.95\linewidth]{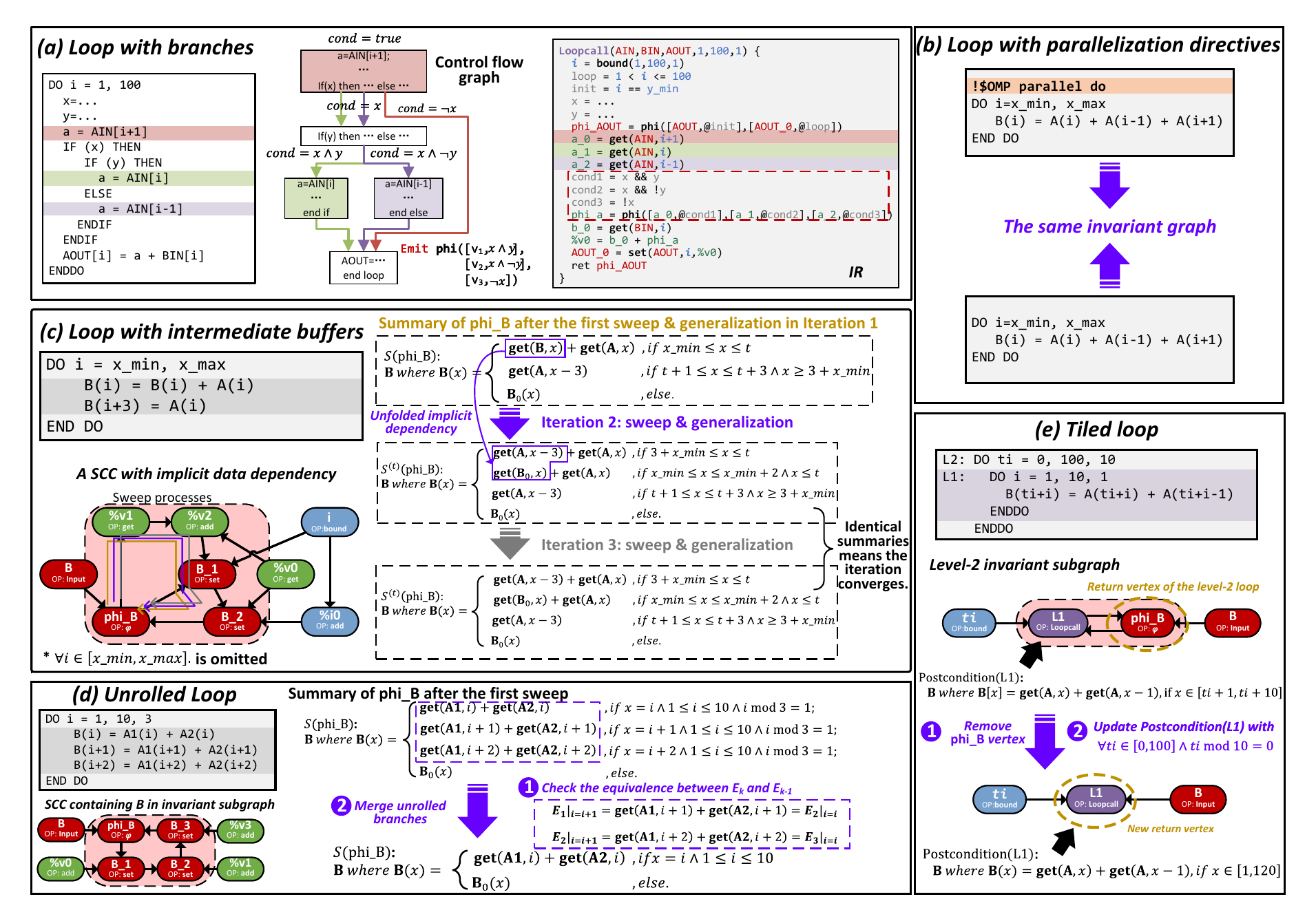}
    \vspace{-0.3cm}
    \caption{Details in Stencil-Lifting for complex cases.}
    \label{fig:case-study}
    \vspace{-0.4cm}
\end{figure}

\subsection{Algorithm Completeness}

From the discussion above, the hierarchical recursive lifting algorithm could derive the stencil summary under two conditions: (1) the successful construction of SMG and the corresponding invariant graph for the target stencil code, and (2) the successful traversal and summarization of the hierarchical recursive lifting process over all levels of the invariant graph. These two conditions guide our analysis of the algorithm completeness.

\stitle{From a stencil code to its invariant graph.}
As discussed in Section~\ref{section: Formalization of Loops in Stencil Code}, Stencil-Lifting focuses on stencil codes implemented as $K$-level nested loops, as this structure is crucial for abstracting the repetitive computation patterns inherent in stencil computations. Although a stencil kernel could, in principle, be implemented using an arbitrary mesh-point update loop, such implementations result in random memory access and computation update across mesh dimensions, thereby making it difficult to extract and summarize invariant computation patterns. In contrast, a $K$-level nested loop naturally gives rise to a $K$-level SMG, which hierarchically captures data dependencies at each loop level (Section~\ref{section: Multilevel Data Dependency Graph}). This hierarchical structure enables the identification and abstraction of invariant computation patterns at each loop level. As a result, repetitive computation behaviors within nested loops can be systematically functionalized in a level-by-level manner (Observation~2), forming the foundation of the recursive lifting algorithm.

Furthermore, deriving invariant subgraphs $G^*_k(V^*_k, E^*_k)$ from SMG is equivalent to abstracting and representing the invariant computation behavior of each loop $L^{(k)}$. While an SMG can be directly constructed from a nested loop, invariant subgraphs may not exist if the loop lacks regularity or consistent computation patterns. Fortunately, stencil kernels inherently exhibit such invariance (Observation 1) and are typically implemented with regular loops, where the same statements are executed repeatedly with minimal branching. Section~\ref{section: Invariant Graph and Summary Attribute} formalizes the notion of regular loops (Definition~\ref{definition: regular loop}) and describes how invariant subgraphs are constructed from them, enabling the formation of an invariant graph. 

By summarizing the entire derivation, we obtain the following result.

\begin{lemma}
\label{lemma: the establishment of SMG and IG}
Assume that a stencil kernel is implemented as a $K$-level nested loop $L$ (Definition \ref{definition: Nested Loop}), and each loop $L^{(k)}$ at level-$k$ is regular (as Definition \ref{definition: regular loop}). Then, a $K$-level SMG can be constructed for the stencil kernel, and a corresponding $K$-level invariant graph can be derived by abstracting the invariant computation pattern at each loop level.
\end{lemma}

\etitle{Remark.} Lemma \ref{lemma: the establishment of SMG and IG} guarantees the existence and constructibility of an invariant graph for stencil kernels implemented as regular nested loops.

\stitle{Hierarchical recursive lifting over the invariant graph.}
In Section~\ref{section: Summary Lifting}, the hierarchical recursive lifting algorithm is developed under the key precondition that the invariant graph of the stencil kernel can be successfully constructed. Theorem~\ref{theorem: hierarchical recursive lifting} reveals that, given the invariant graph, the algorithm derives a sound summary that reflects the stencil's computation behavior.

\begin{lemma}
\label{lemma: soundness of the hierarchical recursive lifting}
For a stencil kernel, if its invariant graph is can be successfully constructed, then the hierarchical recursive lifting algorithm can accurately derive its summary.
\end{lemma}
\etitle{Remark.} Lemma \ref{lemma: soundness of the hierarchical recursive lifting} is a direct corollary of Theorem \ref{theorem: hierarchical recursive lifting}.

\stitle{Completeness.} 
The completeness of the hierarchical recursive lifting algorithm follows directly from the preceding two lemmas.

\begin{theorem}[Algorithm Completeness]
\label{theorem: algorithm completeness}
Let a stencil kernel be implemented as a $K$-level nested loop $L$ (Definition \ref{definition: Nested Loop}), where each loop $L^{(k)}$ at level-$k$ is regular (as defined in Definition \ref{definition: regular loop}). Then, the hierarchical recursive lifting algorithm is complete: it derives a sound summary of the stencil kernel.
\end{theorem}

\etitle{Remark.} Although not all stencil codes are guaranteed to follow a regular nested-loop structure, we observe that stencil kernels in standard benchmarks, prior work, and practical applications typically exhibit such regularity, even in the presence of non-trivial optimizations.

\section{Evaluation}
\label{section: Evaluation}

\subsection{Experimental Setup}

\stitle{Platforms.} The experiments are conducted on a server that consists of two Intel Xeon Gold 6226R CPUs (total 2*16 cores, hyper-threading on), 256GB RAM, and eight NVIDIA Tesla V100 GPUs. The software environment of the server is configured with CUDA v11.7 and Ubuntu Linux 20.04.5 LTS (kernel version 5.15.0-48-generic, GCC 9.4.0, GNU Fortran 9.4.0, LLVM 19.1.5, Halide 19.0.0). Halide code is autotuned to find suitable schedules using the OpenTuner framework \cite{Ansel2014}. In the experiments, STNG \cite{Kamil2016} and Dexter \cite{Ahmad2019} are used as the baselines, which lift stencils by CEGIS. The evaluation of Stencil-Lifting is on the following benchmarks.

\stitle{Two stencil benchmark suites.} We use two benchmark suites: \emph{CloverLeaf} and \emph{StencilMark}. \emph{CloverLeaf} is a standard benchmark \cite{mallinson2013} that consists of multiple Lagrangian-Eulerian hydrodynamics codes using an explicit second-order method on a Cartesian grid. As a part of the Mantevo mini-apps project \cite{Crozier2009}, some codes aim to solve the compressible Euler equations on a 2D staggered grid. \emph{StencilMark} (SM) is a set of stencil microbenchmarks \cite{Kamil2013} for performance and compiler experimentation. We choose three 3D kernels in the suite, whose Fortran codes are open source.

\stitle{Four scientific applications.} We select \emph{WRF}, \emph{NAS MG}, \emph{Heat} and \emph{MagEmul} as application benchmarks. The  \emph{WRF} model \cite{WRF2022} is a state-of-the-art mesoscale numerical weather prediction system consisting of multiple representative stencil Fortran codes. For \emph{WRF} model, we chose eighteen complicated kernels that contain at least two non-trivial optimizations. \emph{NAS MG} is a well-known benchmark \cite{Bailey1991} that applies the V-cycle multigrid method to solve a discrete Poisson equation in 3D. \emph{Heat} is a popular and complex program \cite{epperson2013introduction} drawn from StencilProbe \cite{StencilProbe2025}, which models thermal radiation and diffusion processes. Besides public benchmarks, we have also applied Stencil-Lifting to \emph{MagEmul} \cite{LiHanyu2014}, an electromagnetic field emulation program written in Fortran. The stencil kernels in \emph{MagEmul}, implementing Godunov methods, are more complicated, featuring imperfect loop nesting and branches in loop bodies.

\begin{table}[]
    \scriptsize 
        \begin{center}
        \renewcommand\arraystretch{1.25}
        \caption{Comparison between STNG and Stencil-Lifting. For all the kernels, two common optimization (Opt.) techniques are applied: parallelization directives (PD) for loops and the utilization of intermediate buffers (IB). The mark ``--'' indicates that the corresponding kernel is not optimized.}
        \vspace{-1em}
        \resizebox{\linewidth}{!}{
        \begin{tabular}{c|c|c|c|c|c|c|c|c|c|c|c|c|c}
        \toprule[1.2pt]
            \textbf{\rotatebox[origin=c]{90}{Benchmark}} & 
            \textbf{Kernel} & 
            \textbf{\begin{tabular}[c]{@{}c@{}}Fortran\\ Time\\ (ms)\end{tabular}} & 
            \textbf{\begin{tabular}[c]{@{}c@{}}Halide \\CPU \\Time \\(STNG) \\(ms)\end{tabular}} & 
            \textbf{\begin{tabular}[c]{@{}c@{}}Halide \\CPU \\Time \\(Stencil- \\ Lifting) \\(ms)\end{tabular}} & 
            \textbf{\begin{tabular}[c]{@{}c@{}}Halide\\CPU\\Time\\Speedup\end{tabular}} & 
            \textbf{\begin{tabular}[c]{@{}c@{}}Halide \\GPU\\Time \\ (STNG)\\(ms)\end{tabular}} &
            \textbf{\begin{tabular}[c]{@{}c@{}}Halide \\GPU \\Time \\ (Stencil- \\ Lifting) \\(ms)\end{tabular}} &
            \textbf{\begin{tabular}[c]{@{}c@{}}Halide\\GPU\\Time\\Speedup\end{tabular}} &
            \textbf{\begin{tabular}[c]{@{}c@{}}Lifting\\ Time \\ of \\ STNG  \\(s)\end{tabular}} &
            \textbf{\begin{tabular}[c]{@{}c@{}}Lifting\\ Time of \\ Stencil-\\ Lifting \\ (s)\end{tabular}} &
            \textbf{\begin{tabular}[c]{@{}c@{}}Stencil-\\ Lifting \\ Speedup \\ over\\ STNG\end{tabular}} &
            \textbf{\begin{tabular}[c]{@{}c@{}}Dimen.\\-Points\\-Outputs \end{tabular}} &
            \textbf{\begin{tabular}[c]{@{}c@{}}Opt.\end{tabular}} \\
        \hline
            \multirow{40}{*}{\rotatebox{90}{\textbf{Cloverleaf}}} & akl81 & 1792 & 240.86 & 185.35 & 1.30 & 164.71 & 131.95 & 1.25 & 14944 & 9.40 & 1590.13 & 2d-4p-1o &  PD,IB\\
            \cline{2-14} & akl83 & 3553.40 & 777.55 & 256.84 & 3.03 & 429.67 & 165.68 & 2.59 & 683 & 5.56 & 122.82 & 2d-4p-1o & PD\\
            \cline{2-14} & akl84 & 2798.40 & 620.49 & 272.30 & 2.28 & 355.13 & 51.57 & 6.89 & 631 & 5.59 & 112.9 & 2d-4p-1o & PD\\
            \cline{2-14} & akl85 & 2484.20 & 613.38 & 251.40 & 2.44 & 362.66 & 46.01 & 7.88 & 662 & 7.72 & 85.73 & 2d-4p-1o & PD \\
            \cline{2-14} & akl86 & 2639.60 & 653.37 & 234.03 & 2.79 & 439.93 & 45.93 & 9.58 & 919 & 9.36 & 98.24 & 2d-4p-1o & PD \\
            \cline{2-14} & ackl95 & 2591.20 & 618.42 & 308.26 & 2.01 & 346.88 & 46.19 & 7.51 & 1088 & 5.77 & 188.63 & 2d-4p-2o & PD\\
            \cline{2-14} & amkl100 & 2970.80 & 773.65 & 285.12 & 2.71 & 442.08 & 52.41 & 8.44 & 1512 & 6.53 & 231.55 & 2d-2p-2o &PD \\
            \cline{2-14} & amkl101 & 1283.80 & 352.69 & 128.43 & 2.75 & 235.99 & 50.14 & 4.71 & 1273 & 3.99 & 319.05 & 2d-4p-1o & PD\\
            \cline{2-14} & amkl103 & 1239.87 & 362.54 & 207.13 & 1.75 & 183.14 & 28.43 & 6.44 & 176 & 3.25 & 54.14 & 2d-5p-1o &PD \\
            \cline{2-14} & amkl105 & 1659.60 & 492.46 & 126.95 & 3.88 & 321.01 & 67.04 & 4.79 & 707 & 4.10 & 172.65 & 2d-4p-1o  & PD\\
            \cline{2-14} & amkl107 & 1028.65 & 260.42 & 231.56 & 1.12 & 118.10 & 25.02 & 4.72 & 133 & 5.520 & 24.1 & 2d-5p-1o  &PD \\
            \cline{2-14} & amkl97 & 3294 & 786.16 & 374.43 & 2.10 & 471.92 & 123.50 & 3.82 & 4099 & 8.49 & 483.03 & 2d-3p-2o & PD \\
            \cline{2-14} & amkl98 & 3243.20 & 599.48 & 323.95 & 1.85 & 455.51 & 122.32 & 3.72 & 4191 & 8.43 & 496.98 & 2d-3p-2o & PD\\
            \cline{2-14} & amkl99 & 3001.40 & 723.23 & 327.83 & 2.21 & 403.41 & 53.13 & 7.59 & 1736 & 6.58 & 263.99 & 2d-2p-2o & PD\\
            \cline{2-14} & fckl89 & 2681.80 & 564.59 & 213.55 & 2.64 & 283.19 & 152.18 & 1.86 & 425 & 4.62 & 92.05 & 2d-2p-1o & PD\\
            \cline{2-14} & fckl90 & 2694.40 & 758.99 & 182.70 & 4.15 & 408.24 & 60.71 & 6.72 & 131 & 4.76 & 27.5 & 2d-2p-1o & PD\\
            \cline{2-14} & gckl77 & 630 & 249.01 & 73.60 & 3.38 & 96.48 & 29.24 & 3.30 & 7 & 3.53 & 1.98 & 2d-1p-1o &-- \\
            \cline{2-14} & gckl78 & 658 & 180.27 & 78.90 & 2.28 & 116.46 & 25.60 & 4.55 & 8 & 3.37 & 2.37 & 2d-1p-1o  & --\\
            \cline{2-14} & gckl79 & 636.60 & 168.86 & 76.97 & 2.19 & 96.60 & 29.46 & 3.28 & 6 & 2.20 & 2.73 & 2d-1p-1o  & --\\
            \cline{2-14} & gckl80 & 648.20 & 252.22 & 86.12 & 2.93 & 112.15 & 25.46 & 4.40 & 7 & 2.27 & 3.09 & 2d-1p-1o  & --\\
            \cline{2-14} & ickl10 & 226.84 & 67.31 & 18.74 & 3.59 & 0.18 & 0.18 & 1.00 & 2 & 1.97 & 1.01 & 1d-2p-1o & --\\
            \cline{2-14} & ickl11 & 191.39 & 46.79 & 13.20 & 3.54 & 0.34 & 0.31 & 1.1 & 1 & 1.73 & 0.58 & 1d-1p-1o  &-- \\
            \cline{2-14} & ickl12 & 272.22 & 44.48 & 19.54 & 2.28 & 0.29 & 0.28 & 1.04 & 1 & 2.09 & 0.48 & 1d-2p-1o  &-- \\
            \cline{2-14} & ickl13 & 192.63 & 49.52 & 14.14 & 3.5 & 0.13 & 0.12 & 1.08 & 1 & 1.73 & 0.58 & 1d-1p-1o &-- \\
            \cline{2-14} & ickl14 & 416.49 & 134.35 & 129.33 & 1.04 & 74.51 & 53.25 & 1.40 & 5 & 5.26 & 0.95 & 2d-2p-1o  & --\\
            \cline{2-14} & ickl15 & 403.96 & 151.30 & 148.31 & 1.02 & 60.56 & 50.39 & 1.20 & 13 & 4.89 & 2.66 & 2d-1p-1o & --\\
            \cline{2-14} & ickl16 & 563.33 & 199.76 & 128.17 & 1.56 & 83.09 & 76.69 & 1.08 & 7 & 5.03 & 1.39 & 2d-1p-1o  & --\\
            \cline{2-14} & ickl8 & 143.89 & 35.1 & 16.87 & 2.08 & 0.11 & 0.11 & 1.00 & 1 & 1.92 & 0.52 & 1d-1p-1o  &-- \\
            \cline{2-14} & ickl9 & 142.77 & 34.74 & 17.19 & 2.02 & 0.20 & 0.18 & 1.11 & 1 & 1.74 & 0.58 & 1d-1p-1o  & --\\
            \cline{2-14} & rfkl109 & 404.66 & 116.62 & 115.23 & 1.01 & 59.51 & 48.36 & 1.23 & 6 & 5.14 & 1.17 & 2d-1p-1o  & --\\
            \cline{2-14} & rfkl110 & 430.16 & 150.93 & 135.17 & 1.12 & 61.80 & 57.76 & 1.07 & 5 & 5.19 & 0.96 & 2d-1p-1o  &-- \\
            \cline{2-14} & rfkl111 & 429.96 & 143.32 & 142.60 & 1.01 & 63.42 & 41.66 & 1.52 & 6 & 5.09 & 1.18 & 2d-1p-1o  & --\\
            \cline{2-14} & rfkl112 & 425.18 & 149.19 & 113.09 & 1.32 & 55.15 & 42.39 & 1.30 & 7 & 5.15 & 1.36 & 2d-1p-1o & --\\
            \cline{2-14} & ackl91 & 2774.60 & 534.61 & 380.18 & 1.41 & 336.72 & 89.92 & 3.74 & 6361 & 10.63 & 598.29 & 2d-3p-2o & PD,IB\\
            \cline{2-14} & ackl92 & 2607.20 & 544.30 & 267.76 & 2.03 & 332.55 & 44.06 & 7.55 & 1542 & 5.82 & 265.09 & 2d-2p-2o &PD\\
            \cline{2-14} & ackl94 & 3374.60 & 707.46 & 376.44 & 1.88 & 468.69 & 44.91 & 10.44 & 4273 & 10.17 & 420.12 & 2d-3p-2o &PD,IB \\
            \cline{2-14} & ackl102 & 6221.20 & 1180.49 & 1163.45 & 1.01 & 1009.94 & 72.91 & 13.85 & 9153 & 38.42 & 238.22 & 2d-2p-6o &PD,IB \\
            \cline{2-14} & ackl106 & 2019.22 & 430.54 & 208.81 & 2.06 & 81.09 & 40.48 & 2 & 5546 & 38.38 & 144.5 & 2d-2p-6o &PD,IB \\
            \cline{2-14} & rkl87 & 403.33 & 134.89 & 121.77 & 1.11 & 53.49 & 49.73 & 1.08 & 5 & 5.13 & 0.98 & 2d-1p-1o & --\\
            \cline{2-14} & rkl88 & 436.12 & 143.46 & 142.31 & 1.01 & 65.39 & 55.68 & 1.17 & 5 & 5.09 & 0.98 & 2d-1p-1o  &-- \\
            \hline
            \multirow{3}{*}{\rotatebox{90}{\textbf{NAS MG}}} & mgl15 & 1298.02 & 228.12 & 161.85 & 1.41 & 81.38 & 59.75 & 1.36 & 8 & 8.06 & 0.99 & 3d-1p-1o &-- \\
            \cline{2-14} & mgl18 & 1618.47 & 274.78 & 152.41 & 1.80 & 121.05 & 68.46 & 1.77 & 7 & 8.16 & 0.86 & 3d-1p-1o  & --\\
            \cline{2-14} & mgl5  & 6257.94 & 357.39 & 355.12 & 1.01 & 1455.33 & 123.15 & 11.82 & 46422 & 19.44 & 2387.96 & 3d-19p-1o  & PD,IB\\
            \hline
            \multirow{3}{*}{\rotatebox{90}{\textbf{SM}}} & div0 & 10127.40 & 1045.14 & 745.87 & 1.40 & 2707.86 & 1398.93 & 1.94 & 6590 & 12.35 & 533.58 & 3d-7p-1o  & PD\\
            \cline{2-14} & heat0 & 8805.12 & 779.21 & 746.62 & 1.04 & 2305.01 & 1724.62 & 1.34 & 4668 & 12.15 & 384.25 & 3d-7p-1o &PD \\
            \cline{2-14} & grad0 & 11809.20 & 1819.60 & 729.51 & 2.49 & 1144.30 & 1124.82 & 1.02 & 18557 & 24.66 & 752.47 & 3d-7p-3o &PD \\
        \bottomrule[1.2pt]
        \end{tabular}
        }
        \label{table: Overall lifting results}
        \end{center}
        \vspace{-3em}
    \end{table}

\subsection{Comparison with State-of-the-art Systems}

\subsubsection{Lifting Effectiveness and Efficiency}

Table~\ref{table: Overall lifting results} reports the lifting performance on all stencil kernels successfully processed by STNG, including all the open-source benchmarks evaluated in  STNG's original work~\cite{Kamil2016}. Due to the lack of active maintenance of STNG’s public repository, it is difficult to apply it to newer or more complex stencil kernels, limiting the scope of direct comparison. In this evaluation, both STNG and Stencil-Lifting are executed in single-threaded mode. Compared to STNG, Stencil-Lifting achieves a geometric mean speedup of 31.62$\times$ in lifting time and exhibits more consistent runtime across similar kernels. This improvement is largely due to replacing CEGIS-based synthesis engine with a deterministic, iteration-based lifting strategy. Table~\ref{tab:lift-dexter} further compares the lifting time among STNG, Dexter, and Stencil-Lifting. While Dexter reduces synthesis complexity relative to STNG, Stencil-Lifting achieves a 5.8$\times$ speedup over Dexter on average, confirming its superior effectiveness and efficiency.

\subsubsection{Performance of Generated Codes}

Table~\ref{table: Overall lifting results} also compares the original Fortran codes with their corresponding Halide versions generated from the lifted summaries. For the CPU backend, Column 3 lists the single-threaded execution time of the original code. Columns 4 and 5 show the execution times of CPU-targeted Halide code generated by STNG and Stencil-Lifting, respectively. The Halide codes generated by Stencil-Lifting achieve 7.62$\times$ and 1.88$\times$ geometric mean speedups against the Fortran baselines and the codes generated by STNG. Stencil-Lifting consistently generates more optimized code due to its fine-grained autotuning schedule. For the GPU backend, Columns 7 and 8 report the execution time of the GPU-targeted Halide code generated by STNG and Stencil-Lifting, respectively (excluding CPU-GPU data transfer time). The GPU implementations generated by Stencil-Lifting demonstrate 83.85$\times$ and 2.72$\times$ geometric mean speedups over the Fortran baselines and GPU codes generated by STNG.

\begin{table}[]
    \centering
    \scriptsize
    \renewcommand\arraystretch{1.2}
    \begin{minipage}{.52\linewidth}
        \centering
    \begin{minipage}{.9\linewidth}
        \caption{Comparison of lifting efficiency among STNG, Dexter, and Stencil-Lifting. Each kernel is selected from Table \ref{table: Overall lifting results}.}
        \vspace{-1em}
        \begin{tabularx}{\linewidth}{c|c|c|c|c}
        \toprule[1.2pt]
            \textbf{Kernel} &
            \textbf{\begin{tabular}[c]{@{}c@{}}Dimen.\\ and\\Points\end{tabular}} &
            \textbf{\begin{tabular}[c]{@{}c@{}} STNG \\ Lifting\\ Time (s)\end{tabular}} &
            \textbf{\begin{tabular}[c]{@{}c@{}} Dexter \\ Lifting\\ Time (s)\end{tabular}} &
            \textbf{\begin{tabular}[c]{@{}c@{}}Stencil-Lifting\\ Lifting\\Time (s)\end{tabular}} \\
            \midrule[1pt]
            akl81 & 2d-4p & 14944 & 104.8 & 9.4 \\ 
            ackl91 & 2d-3p & 6361 & 58.1 & 10.6 \\ 
            amkl97 & 2d-3p & 4099 & 35.5 & 8.5 \\ 
            grad0 & 3d-7p & 18557 & 148.3 & 24.7 \\
        \bottomrule[1.pt]
        \end{tabularx}
        \label{tab:lift-dexter}
    \end{minipage}
    \end{minipage}
    \hfill
    \begin{minipage}{0.47\linewidth}
        \centering
        \includegraphics[width=\linewidth]{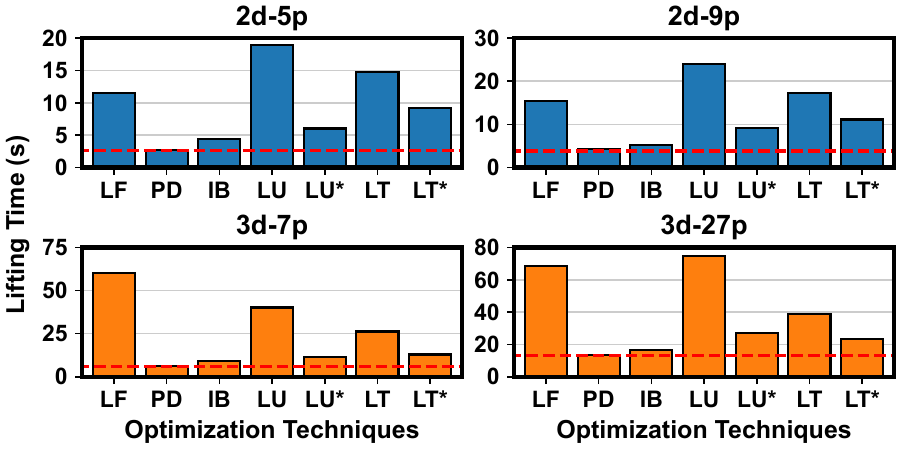}
        \vspace{-2.5em}
    \captionof{figure}{Lifting time of Stencil-Lifting for hand-optimized stencils. LU$^*$ and LT$^*$ mean that the \textit{equivalence checking} and \textit{vertex elimination} are used.}
    \label{fig:Ablation}
    \end{minipage}
\end{table}

\begin{table}[]
\scriptsize 
    \begin{center}
    \caption{Practicality evaluation on complex real-world stencil kernels. For each kernel, the applied optimizations are drawn from five techniques: (1) loop fusion (LF) for boundary conditions, (2) parallelization directives (PD) for loops, (3) the use of intermediate buffers (IB), (4) loop unrolling (LU), and (5) loop tiling (LT).}
    \vspace{-1em}
    \renewcommand\arraystretch{1.25}
    \resizebox{\linewidth}{!}{
    \begin{tabular}{c|c|c|c|c|c|c|c|c|c}
    \toprule[1.2pt]
        \textbf{\rotatebox[origin=c]{90}{App.}} & 
        \textbf{Kernel} & 
        \textbf{\begin{tabular}[c]{@{}c@{}}Dimen.\\-Points\\-Outputs \end{tabular}} &
        \textbf{\begin{tabular}[c]{@{}c@{}}Optimize\\ Tech.\end{tabular}} &
        \textbf{\begin{tabular}[c]{@{}c@{}}Lifting \\ Time \\ (s)\end{tabular}} &
        \textbf{\begin{tabular}[c]{@{}c@{}}Fortran \\ Time (ms)\end{tabular}} & 
        \textbf{\begin{tabular}[c]{@{}c@{}}Halide \\ CPU \\ Time (ms)\end{tabular}} & 
        \textbf{\begin{tabular}[c]{@{}c@{}}Halide \\ CPU \\ Speedup\end{tabular}} & 
        \textbf{\begin{tabular}[c]{@{}c@{}}Halide \\ GPU \\ Time (ms)\end{tabular}} & 
        \textbf{\begin{tabular}[c]{@{}c@{}}Halide \\ GPU \\ Speedup\end{tabular}} \\
    \hline
        \multirow{11}{*}{\rotatebox{90}{\textbf{Cloverleaf}}} & ackl93    & 2d-3p-2o & LF,PD,IB & 17.37 & 4643.04  & 340.49 & 13.63 & 45.66  & 101.68 \\
        \cline{2-10}                 & ackl96    & 2d-3p-2o & LF,PD,IB & 17.4  & 4385.97  & 578.46 & 7.58  & 57.89  & 75.76 \\ 
        \cline{2-10}                 & amkl102   & 2d-4p-2o & LF,PD    & 8.28  & 1957.14  & 545.46 & 3.59  & 53.96  & 36.27 \\ 
        \cline{2-10}                 & amkl104   & 2d-3p-1o  & LF,PD,IB & 7.48  & 4825.55  & 161.13 & 29.95 & 40.62  & 118.8 \\ 
        \cline{2-10}                 & amkl106   & 2d-4p-2o & LF,PD    & 7.85  & 2747.77  & 566.55 & 4.85  & 37.18  & 73.9 \\ 
        \cline{2-10}                 & amkl108   & 2d-3p-1o  & LF,PD,IB & 7.13  & 2870.23  & 254.22 & 11.29 & 27.88  & 102.95 \\ 
        \cline{2-10}                 & cdkl      & 2d-4p-1o & LF,PD,IB & 11.96 & 5310.42  & 488.29 & 10.88 & 37.6   & 141.23 \\ 
        \cline{2-10}                 & igkl      & 2d-1p-1o  & LF,PD    & 5.58  & 1341.79  & 271.24 & 4.95  & 109.95 & 12.2 \\ 
        \cline{2-10}                 & pdvk1     & 2d-4p-2o & PD,IB    & 13.33 & 3402.17  & 905.91 & 3.76  & 59.64  & 57.05 \\ 
        \cline{2-10}                 & pdvk2     & 2d-4p-2o & PD,IB    & 14.52 & 2940.13  & 1008.8 & 2.91  & 86.84  & 33.86 \\
        \cline{2-10}                 & viscosity & 2d-4p-1o & LF,PD    & 13.24 & 17689.56 & 400.65 & 44.15 & 126.17 & 140.2 \\
    \hline
        \multirow{18}{*}{\rotatebox{90}{\textbf{WRF}}}   & mdel9  & 3d-1p-1o  & LF,PD    & 11.06   & 956.9   & 55.86  & 17.13  & 34.18 & 28   \\
        \cline{2-10}                   & mdel24 & 3d-4p-2o & LF,PD    & 34.58   & 6486.78 & 64.52  & 100.54 & 42.68 & 151.99 \\
        \cline{2-10}                   & mdel26 & 3d-1p-1o  & LF,PD,IB & 2.50    & 5117.05 & 115.67 & 44.24  & 48.96 & 104.51 \\
        \cline{2-10}                   & mdel27 & 3d-2p-1o  & LF,PD    & 15.63   & 1212.71 & 177.47 & 6.83   & 50.42 & 24.05 \\
        \cline{2-10}                   & mdel29 & 3d-1p-2o  & LF,PD    & 14.72   & 5268.15 & 95.63  & 20.72  & 41.58 & 126.7 \\
        \cline{2-10}                   & mdel30 & 3d-1p-1o  & LF,PD    & 7.95    & 1209.56 & 58.38  & 55.09  & 36.51 & 33.13 \\
        \cline{2-10}                   & mdel33 & 3d-11p-4o & PD,IB    & 27.23   & 4928.4  & 855.02 & 6.83   & 46.04 & 107.05 \\
        \cline{2-10}                   & mdel34 & 3d-2p-1o & LF,PD    & 34.46   & 5657.63 & 143.1  & 39.54  & 56.98 & 99.29 \\
        \cline{2-10}                   & mdel35 & 3d-14p-1o  & PD,IB    & 21.32   & 3132.8  & 410.88 & 46.55  & 31.2  & 100.41 \\
        \cline{2-10}                   & mdel36 & 3d-8p-1o & LF,PD    & 4.63    & 6482.49 & 139.27 & 46.55  & 44.38 & 146.07 \\
        \cline{2-10}                   & mbel10 & 3d-1p-2o  & LF,PD    & 34.63   & 6997.42 & 269.48 & 25.97  & 42.72 & 163.8 \\
        \cline{2-10}                   & mbel16 & 3d-1p-2o  & LF,PD    & 34.33   & 7791.63 & 245.84 & 31.69  & 40.01 & 194.74 \\
        \cline{2-10}                   & miel3  & 3d-4p-2o & LF,PD,IB & 31.01   & 7828.96 & 265.13 & 29.53  & 41.35 & 189.33 \\
        \cline{2-10}                   & miel10 & 3d-3p-1o & PD,IB    & 31.76   & 3240.42 & 243.76 & 13.29  & 98.05 & 33.05 \\
        \cline{2-10}                   & miel11 & 3d-3p-1o & LF,PD,IB & 30.03   & 2808.54 & 215.69 & 13.02  & 45.59 & 61.6 \\
        \cline{2-10}                   & miel12 & 3d-4p-1o & LF,PD,IB & 33.96   & 2978.71 & 252.15 & 11.81  & 61.31 & 48.58 \\
        \cline{2-10}                   & miel13 & 3d-4p-1o & LF,PD,IB & 34.85   & 2968.37 & 286.5  & 10.36  & 34.67 & 85.62 \\
        \cline{2-10}                   & miel14 & 3d-3p-1o & PD,IB    & 34.17   & 2635.64 & 243.53 & 10.82  & 36.84 & 71.54 \\
    \hline
        \multirow{3}{*}{\rotatebox{90}{\textbf{NAS MG}}} & mgl6            & 3d-21p-2o & PD,IB & 19.72  & 2284.57 & 711.56 & 3.21 & 42.06 & 54.32 \\
        \cline{2-10}            & resid1          & 3d-21p-1o & PD,IB & 20.04  & 2429.89 & 745.01 & 3.26 & 46.69 & 52.04 \\
        \cline{2-10}            & trilinear\_proj & 3d-27p-1o & PD,IB & 36.92  & 2492.2  & 779.95 & 3.2  & 36.54 & 68.2 \\
    \hline
        \multirow{2}{*}{\rotatebox{90}{\textbf{Heat}}}  & probe     & 3d-25p-1o & LF,PD,LT & 43.12 & 1812.51 & 68.95 & 26.29 & 8.61 & 210.51 \\
        \cline{2-10}           & multiload & 1d-3p-1o  & PD,LU    & 6.91  & 139.97  & 77.39 & 1.81  & 2.21 & 63.33 \\  
    \hline 
	   \multirow{4}{*}{\rotatebox{90}{\textbf{MagEmul}}} & eupdate  & 3d-10p-9o & PD,IB    & 711.36 & 71056.67 & 21553.38 & 3.3   & 1244.25 & 57.11 \\
        \cline{2-10}             & solvepde & 3d-5p-1o  & PD,IB,LT & 510.53 & 1817.62  & 545.36   & 3.33  & 31.71   & 57.32 \\
        \cline{2-10}             & compgrad & 3d-3p-2o  & LF,PD       & 18.93  & 1175.61  & 143.88   & 8.17  & 10.617  & 110.73 \\
        \cline{2-10}             & getadvrt & 4d-2p-1o  & LF,PD       & 17.04  & 12200.71 & 606.55   & 20.11 & 39.341  & 310.13 \\
    \bottomrule[1.2pt]
    \end{tabular}} 
    \label{table: Application results}
    \end{center}
    \vspace{-3em}
\end{table}

\subsection{Evaluation on Detailed Designs within Stencil-Lifting}

\subsubsection{Performance Evaluation}
Figure~\ref{fig:Ablation} illustrates the performance of Stencil-Lifting on hand-optimized stencil codes that apply non-trivial transformations to classic stencil patterns, including 2d-5p, 2d-9p, 3d-7p, and 3d-27p. The dashed red line marks the lifting time for the corresponding unoptimized baseline kernel. Two key observations emerge from this evaluation. First, Stencil-Lifting successfully handles kernels with five widely-used loop and memory optimizations: (1) loop fusion (LF) for boundary conditions, (2) parallelization directives (PD) for loops, (3) the use of intermediate buffers (IB), (4) loop unrolling (LU), and (5) loop tiling (LT). Although the lifting time increases compared to the unoptimized versions, the lifting process remains robust and scalable. Second, for kernels with unrolling and tiling, the \textit{equivalence checking} and \textit{vertex elimination} procedures incorporated into the hierarchical recursive lifting algorithm significantly reduce redundant computations, achieving an average 2.4$\times$ speedup in lifting time.

\subsubsection{Scalability Evaluation}
In the scalability evaluation, we fix the stencil dimensionality and vary the stencil complexity, defined as the product of the number of computational points and the number of output arrays. Figure~\ref{fig:stat} compares the scaling behavior of STNG and Stencil-Lifting as stencil complexity increases. Stencil-Lifting maintains near-linear scalability, while STNG exhibits rapidly growing lifting time, even for syntactically simple loops. These results highlight the superior scalability of Stencil-Lifting in handling complex stencil computations.

\begin{figure}[btph]
    \vspace{-0.5em}
    \centering
    \includegraphics[width=1\textwidth]{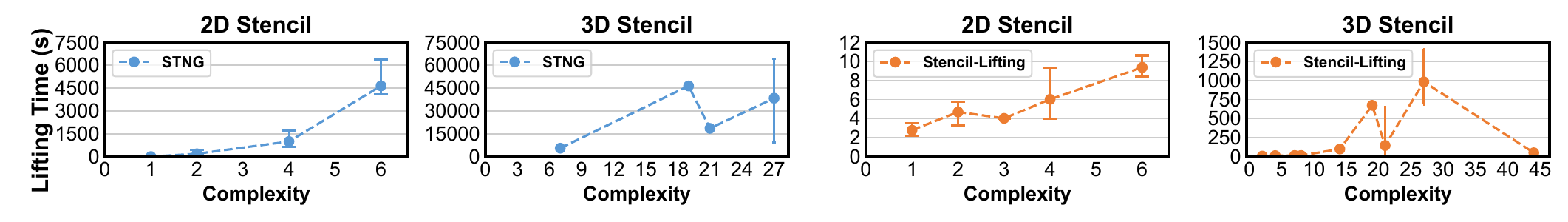}
    \vspace{-1.5em}
    \caption{The summary lifting time of STNG and Stencil-Lifting with the increase of stencil complexity.
    }
    \vspace{-1.5em}
    \label{fig:stat}
\end{figure}

\subsection{Practicality Evaluation on Complex Real-World Stencil Kernels}

Table~\ref{table: Application results} evaluates Stencil-Lifting on $38$ complex stencil kernels drawn from real-world scientific applications. Each kernel features distinct code structures and incorporates at least two non-trivial hand-coded optimization. The five benchmarks cover a broad range of domains, including hydrodynamics (Cloverleaf), climatic simulation (WRF), multigrid computation (NAS MG), thermal transmission (HEAT), and electromagnetic field emulation (MagEmul). They capture diverse stencil coding styles and optimization practices found in both legacy and modern systems.

As shown in Table~\ref{table: Application results}, Stencil-Lifting successfully derives correct summaries for all the kernels. Despite our effort, STNG was unable to complete these cases. In contrast, Stencil-Lifting completes all cases within 1,000 seconds, demonstrating excellent practicality and robustness. Moreover, the lifted summaries are translated into Halide DSL for both CPU and GPU backends. Compared to the original (hand-optimized) stencil implementations, the generated CPU-targeted code achieves a geometric mean speedup of 11.35$\times$ (ranging from 1.81$\times$ to 100.54$\times$), while the GPU-targeted code delivers a geometric mean speedup of 78.19$\times$ (ranging from 12.20$\times$ to 310.13$\times$), demonstrating significant performance gains even over manually tuned implementations.

\section{Related Works}

\stitle{Stencil lifting systems.} Over the past decade, researchers have explored lifting legacy stencil code into high-level forms for DSL-based optimization. Early work like Helium \cite{Mendis2015} inferred Halide kernels from dynamic binary traces, but lacked soundness guarantees and scalability. STNG \cite{Kamil2016} advanced the field by formalizing stencil translation as a syntax-constrained synthesis task using CEGIS, enabling Fortran-to-Halide conversion. However, its use of SKETCH synthesizer\cite{Solar-Lezama2007} incurred high complexity when processing deep nested loops and multi-neighborhood patterns, often requiring hours to verify semantic equivalence for complex cases. To address these inefficiencies, Dexter \cite{Ahmad2019} proposed a three-phase decomposition (ROI-Terms-Expr) to reduced variable complexity, but still struggled with exponential search spaces in high-dimensional stencils. Our system takes a fundamentally different approach: it hierarchically extracts loop invariants from the data dependence graph to construct postconditions that capture stencil semantics. Instead of relying on exhaustive synthesis, it employs incremental invariant induction to ensure both scalability and semantic rigor.

\stitle{Loop summarization.} Summarizing loops with semantic preservation is a long-standing topic in program verification, synthesis, and transformation. Classical methods include inferring loop invariants, constraint-based synthesis \cite{Alur2013, Cheung2013}, and IC3-based model checking \cite{Bradley2011, Cimatti2012}. We focus on a specific class of summaries that admit semantically equivalent representations in functional or declarative models, such as Halide, ensuring tractability and soundness. Related works share similar goals with ours, including STNG \cite{Kamil2016} and Dexter \cite{Ahmad2019}. They adopt counterexample-guided program synthesis to derive invariants for stencil loops. Recent efforts on dataflow-based summarization, including LoopSCC \cite{zhu2024} and other approaches \cite{S-looper2015, Proteus2016, Blicha2022}, have advanced loop summarization, particularly for numerical loops. However, these methods primarily target scalar updates and lack support for multi-dimensional array operations typical in stencil codes, limiting their applicability to complex loop structures. Our approach leverages dataflow analysis for hierarchical summarization of nested loops of stencil kernel and verifies derived summaries as loop postconditions within a Hoare logic framework \cite{Hoare1969}.

\section{Conclusion}
We presented Stencil-Lifting, a system that efficiently transforms Fortran stencil kernels from legacy codes into optimized Halide implementations. By capturing stencil computations as invariant subgraphs, Stencil-Lifting enables consistent summary derivation across nested loops. A provably terminating recursive lifting algorithm further ensures efficient summarization without search-based synthesis overhead. Experimental results show that Stencil-Lifting achieves 31.62$\times$ and 5.8$\times$ speedups compared to STNG and Dexter respectively, while maintaining correctness. Stencil-Lifting also improves maintainability and portability of stencil kernels, facilitating legacy code adaptation to heterogeneous platforms.

\section*{Data-Availability Statement}
The implementation of Stencil-Lifting and the raw data from our evaluation in Section \ref{section: Evaluation} are available online \cite{StencilLifting}. 

\begin{acks}
The authors would like to thank all anonymous reviewers for their valuable comments and helpful suggestions. This work is supported in part by the National Natural Science Foundation of China, under grant numbers 62032023, T2125013, and 62172391.
\end{acks}

\appendix
\section{Intermediate Representation}
\label{section: Intermediate Representation}

In Stencil-Lifting, we design a specialized intermediate representation (IR) to describe invariant subgraphs, incorporating three novel features tailored for recursive lifting: 

\etitle{(1) Hierarchy:} Represent explicitly each loop as a function, establishing isolated analysis contexts for each loop level.

\etitle{(2) Dataflow-Centric:} Treat the array as a primitive data type, enabling explicit tracking of read-write dependencies between array operations. This facilitates precise dataflow analysis and improves dependency resolution.  

\etitle{(3) Control-Flow Reduction:} Simplify complex control-flow by encoding execution paths using branching condition predicates. Definition conflicts across execution paths are resolved through specialized $\varphi$ functions, which merge values based on dominating conditions. This approach enhances the efficiency of dataflow analysis on invariant subgraphs.  

To further streamline dataflow analysis, our IR adopts SSA form. Below, we provide key design specifications, including \textit{region}, \textit{value}, and \textit{operation}.

\stitle{Region.} \textit{Region} is a hierarchical structure that encapsulates a sequence of operations, providing a structured representation of control flow and scope within the parent operation. In our design, a \textit{region} is defined to begin with an \textit{input} operation and terminate with a \textit{return} operation. These operations align with the input and output values of the region's parent operations, thereby establishing data dependencies between nested regions.

\stitle{Value.} The proposed IR employs a highly abstract data type design, focusing on symbolically expressing the logical structure and data flow of programs, while abstracting away concrete execution details. The IR supports four fundamental \textit{Value} types: \textit{index value} (for symbolic memory access), \textit{numerical value} (for symbolic arithmetic operations), \textit{array value} (for symbolic data collections), and \textit{boolean value} (for symbolic logical predicates). 

\stitle{Operation.} An operation is a fundamental computation unit encapsulating specific program semantics in a loop. Each operation owns \textit{input values}, \textit{output values}, and optionally, an associated \textit{region}. 
Not all operations require a \textit{region} for semantic elaboration. 
Accordingly, all operations are classified into Region-bearing operations and Non-region operations, depending on whether they contain a non-empty region. 
Stencil-Lifting contains the following operations.

\stitle{\small \CIRCLE} \textit{Scalar operations} are algebraic manipulations between two data except for array. 

\stitle{\small \CIRCLE}  \textit{Array operations}, such as \texttt{get} and \texttt{set}, are reading a numerical value from an array or writing a value into an array. 

\stitle{\small \CIRCLE}  \textit{Input operations} are special operations that require no input value. Their function is to provide symbols for region arguments, which are either values from the outer region or function parameters.

\stitle{\small \CIRCLE}  \textit{Bound operation} is used to create a loop counter constrained within a given range and incremented by a fixed step size. For example, when the inputs are $begin$, $end$, and $step$, the bound operation generates a loop counter that is restricted to the interval $[begin, ~end]$ and increases by $step$ during each iteration.

\stitle{\small \CIRCLE}  \textit{Loopcall operation} functionalizes a inner loop. This operation is region-bearing opertaions for representing a loop as a function. For each input $v_i$ of \texttt{Loopcall} operation, there is a dual value $v_{i'}$ in the loopcall's region. For output value $v_o$, there also is a dual value $v_{o'}$ in loopcall's region. The loopcall's region represents the semantics of the functionalized loop, based on which the invariant subgraph is constructed.

\stitle{\small \CIRCLE} \textit{$\varphi$ function operations} represent the aggregation of conflict definitions on different execution paths. The input of $\varphi$ function operations is a tuple of numerical and boolean values, which represents the defined value and the branching condition on an execution path. In Stencil-Lifting, $\varphi$ function operation is different from the common $\varphi$ function, because our $\varphi$ function encodes execution paths through branching condition predicates.

\section{Proofs of the Theorem \ref{theorem: self-consistent} and Theorem \ref{theorem: recursive lifting}}
\label{app:B}

\begin{theorem}[Theorem \ref{theorem: self-consistent}]

Let $(S^*(v^k_1), S^*(v^k_2), \cdots, S^*(v^k_{h_k}))$ be a \textit{self-consistent} summary configuration under the $h_k$ rules $R(v^k_1), R(v^k_2), \cdots, R(v^k_{h_k})$. If $S^*(v^k_{p_k})|_{\mathcal{L}(v^k_{p_k})=(m,n)}$ is set as the loop invariant and $S^*(v^k_{p_k})|_{\mathcal{L}(v^k_{p_k})=(m,N_k)}$ is set as the postcondition for $L^{(k)}$ at its $m$-th invocation (where $1\leq m \leq M_k$ and $ 1 \leq n \leq N_k$), then the three statements in Equation (\ref{equation: VC}) hold.
\end{theorem}
\begin{proof}
Define the precondition as $S^*(v^k_{p_k})|_{I(v^k_{p_k})=(m_k,1)}$. Furthermore, according to the definition of $S^*(v^k_{p_k})|_{\mathcal{L}(v^k_{p_k})=(m,n_k)}$ and
$S^*(v^k_{p_k})|_{\mathcal{L}(v^k_{p_k})=(m,N_k)}$, it is easy to check the statements I and III in Equation (\ref{equation: VC}) are valid. In the following, we focus on the verification of statement II in Equation (\ref{equation: VC}).

According to the definition of invariant subgraphs, we can use $G^*_k(V^*_k, E^*_k)$ to generate the dataflow graph $G'(V',E')$ of each loop execution at the $k$-th level of SMG. Denote the summary the copy of $v^*_k \in V^*$ at the $i$-the loop slice as $S_k|_{I_k=i}$. Assume that $G'(V', E')$ has $n_{input}$ input vertices, i.e., $v_1,v_2,\cdots, v_{n_{input}}$ where the iteration property of $v_j$ is $I_j = 1$ and the summary property of $v_j$ is denoted as $S_j$. Further, $v_1,v_2,\cdots, v_{n_{input}}$ also are all the input of the first loop slice $G'(V'_1, E'_1)$ of $G'(V', E')$. According to our assumption that we just use an array to contain the final results of each loop execution, $v^k_{p_k}|_{\mathcal{L}(v^k_{p_k})=(m,1)}$ must be the only output vertex of $G(V_1, E_1)$. We re-denote $v^k_{p_k}|_{\mathcal{L}(v^k_{p_k})=(m,1)}$ as $v_{out1}$ for convenience, and define $G'_1$ as the map from $S_1,S_2,\cdots, S_{n_{input}}$ to the summary $S_{out1}$ of $v_{out1}$ under the given rule configuration. As the $G'(V'_1, E'_1)$ is a DAG and $S_{out1}$ can be determined uniquely by all the rules of vertices in $V'_1$ with $S_1,S_2,\cdots, S_{n_{input}}$ as inputs, the map $G'_1$ is well-defined. In the following, we use the mathematical induction method to prove
\begin{equation}
\label{equation: S*p1}
S(v^k_{p_k})|_{I(v^k_{p_k})=(m,1)} = G'_1 (S_1,S_2,\cdots, S_{n_{input}}).
\end{equation}
We define the depth $d$ of $G'(V'_1, E'_1)$ as the number of edges in the longest path from one input vertex to the output vertex $v_{out1}$. When $d=1$, $G'_1 (S_1,S_2,\cdots, S_{n_{input}}) = R(v_{out1})(S_1,S_2,\cdots, S_{n_{input}})$. As $(S^*(v^k_1), S^*(v^k_2), \cdots, S^*(v^k_{h_k}))$ is {\itshape self-consistent}, we have
$$S(v^k_{p_k})|_{I(v^k_{p_k})=(m,1)} = R(v^k_{p_k})(S_1,S_2,\cdots, S_{n_{input}}) = R_{out1}(S_1,S_2,\cdots, S_{n_{input}}).$$ 
Hence, the equality (\ref{equation: S*p1}) is valid for the case with $d=1$.  

Next, assume that the equality (\ref{equation: S*p1}) is valid for all the cases with $d \leq h$. When $d=h+1$, we denote all the parent vertices of $v^k_{p_k}$ as $v^k_{m_1}, v^k_{m_2}, \cdots, v^k_{m_{p_k}}$ in $G^*(V^*_k, E^*_k)$. As 
$S(v^k_{m_1}), S(v^k_{m_2}), \cdots, S(v^k_{m_{p_k}})$ and $S(v^k_(p_k)$ is {\itshape self-consistent} under the rule $R(v^k_{p_k})$, we have 
\begin{equation}
\label{equation: S*p}
S(v^k_{p_k}) = R(v^k_{p_k})(S(v^k_{m_1}), S(v^k_{m_2}), \cdots, S(v^k_{m_{p_k}})).    
\end{equation}
By setting the iteration index as $1$ in the equation (\ref{equation: S*p}), we deduce 
\begin{equation}
\label{equation: S*p11}
S(v^k_{p_k})|_{\mathcal{L}(v^k_{p_k})=(m,1)} = R(v^k_{p_k})(S(v^k_{m_1})|_{\mathcal{L}(v^k_{m_1})=(m,1)}, S(v^k_{m_2})|_{\mathcal{L}(v^k_{m_2})=(m,1)}, \cdots, S(v^k_{m_{p_k}})|_{\mathcal{L}(v^k_{m_{p_k}})=(m,1)}).    
\end{equation}
We denote $G''(V''_j, E''_j)$ as a sub-graph of $G'(V'_1, E'_1)$, which owns $v^k_{m_{j}}|_{\mathcal{L}(v^k_{m_j})=(m,1)}$ as its output vertex ($j=1,2,\cdots, m_{p_k}$). It is clear that the depth of $G''(V''_k, E''_k)$ is not larger than $h$. Denote $G''_j$ as the map from $S_1,S_2,\cdots, S_{n_{input}}$ to $S(v^k_{m_j})|_{\mathcal{L}(v^k_{m_j})=(m,1)}$. Based on the inductive assumption, $S(v^k_{m_j})|_{\mathcal{L}(v^k_{m_j})=(m,1)} = G''_{j} (S_1,S_2,\cdots, S_{n_{input}})$, which together with (\ref{equation: S*p11}) leads to
\begin{equation}
\begin{aligned}
S(v^k_{p_k})|_{I(v^k_{p_k})=(m,1)} 
 & =R(v^k_{p_k})(S(v^k_{m_1})|_{\mathcal{L}(v^k_{m_1})=(m,1)}, S(v^k_{m_2})|_{\mathcal{L}(v^k_{m_2})=(m,1)}, \cdots, S(v^k_{m_{p_k}})|_{\mathcal{L}(v^k_{m_{p_k}})=(m,1)}) \\ 
 & = R_{out1}(G''_{1} (S_1,S_2,\cdots, S_{n_{input}}), \cdots, G''_{m} (S_1,S_2,\cdots, S_{n_{input}})) \\
 & =G'_1 (S_1,S_2,\cdots, S_{n_{input}}).
\end{aligned}   
\end{equation}
According to the mathematical induction method, the equality (\ref{equation: S*p1}) is valid for all the cases.

For any $i \geq 2$, $v_1, v_2,\cdots, v_{n_{input}}$ and $S(v^k_{p_k})|_{I(v^k_{p_k})=(m,i-1)}$ may be the input vertices of $G'(V'_i,E'_i)$. Similarly, we can prove that
\begin{equation}
\label{equation: S*pk}
S(v^k_{p_k})|_{I(v^k_{p_k})=(m,i)} = G_k (S_1,S_2,\cdots, S_{n_{input}}, S(v^k_{p_k})|_{I(v^k_{p_k})=(m,i-1)}),
\end{equation}
is valid for $k \geq 2$. Consequently, together (\ref{equation: S*p1}) and (\ref{equation: S*pk}), we obtain
\begin{equation}
\label{equaiton: S*p-all}
\begin{cases}
S(v^k_{p_k})|_{I(v^k_{p_k})=(m,1)} = G'_1 (S_1,S_2,\cdots, S_{n_{input}}), \\
S(v^k_{p_k})|_{I(v^k_{p_k})=(m,i)} = G'_i (S_1,S_2,\cdots, S_{n_{input}}, S(v^k_{p_k})|_{I(v^k_{p_k})=(m,i-1)}) ~\hbox{where}~ i \geq 2.
\end{cases}
\end{equation}
If we set the loop invariant as $S(v^k_{p_k})$, the loop condition as the iteration index not larger than $N_k$, and the loop body as $G'_i$, the equality (\ref{equaiton: S*p-all}) leads to the statement II in Equation (\ref{equation: VC}).
\end{proof}

\begin{theorem}[Theorem \ref{theorem: recursive lifting}]

For any invariant subgraph \( G^*_k(V^*_k,E^*_k) \) of \( G^*(V^*,E^*) \), let the summary configuration $(S^*(v^k_1), \cdots, S^*(v^k_{h_k}))$ be \textit{self-consistent} under rules $R(v^k_1), \cdots, R(v^k_{h_k})$. Suppose that
$v^k_{q_k}$ and $v^{k-1}_{p_{k-1}}$ satisfy $S^*(v^k_{q_k})|_{\mathcal{L}(v^k_{q_k})=(m_k,n_k)} = S^*(v^{k-1}_{p_{k-1}}) |_{\mathcal{L}(v^{k-1}_{p_{k-1}})=( ((m_k-1) \cdot N_{k} + n_{k}) \cdot N_{k-1},N_{k-1})}$ for any $m_k$ and $n_k$ ($1 \leq m_k \leq M_k$ and $1 \leq n_k \leq N_k$). $S^*(v^K_{p_K})|_{\mathcal{L}(v^K_{p_K})=(M_K,N_K)}$ could represent the summary/postcondition of the stencil kernel.

\end{theorem}

\begin{proof}
For $G^*(V^*,E^*)$, let a summary configuration $(S^*(v^1_1), S^*(v^1_2), \cdots)$ consist of all the summaries $S^*(v^k_i)$ where $1 \leq i \leq h_k$ and $1 \leq k \leq K$. We prove Theorem \ref{theorem: recursive lifting} by two steps.

\stitle{Step 1: Prove the summary configuration $(S^*(v^1_1), S^*(v^1_2), \cdots )$ is self-consists over $G^*(V^*,E^*)$.}

On the one hand, in any Level $k$, the summary configuration $(S^*(v^k_1), S^*(v^k_2), \cdots, S^*(v^k_{h_k}))$ is \textit{self-consistent} under $h_k$ rules $R(v^k_1), R(v^k_2), \cdots, R(v^k_{h_k})$, according to the condition of Theorem \ref{theorem: recursive lifting}. 

On the other hand, for two adjacent levels such as Level $k$ and Level $k-1$, 
$$S^*(v^k_{q_k})|_{\mathcal{L}(v^k_{q_k})=(m_k,n_k)} = S^*(v^{k-1}_{p_{k-1}}) |_{\mathcal{L}(v^{k-1}_{p_{k-1}})=( ((m_k-1) \cdot N_{k} + n_{k}) \cdot N_{k-1},N_{k-1})},$$
according to the condition of Theorem \ref{theorem: recursive lifting}. Based on Theorem \ref{theorem: self-consistent}, $S^*(v^{k-1}_{p_{k-1}})|_{\mathcal{L}(v^{k-1}_{p_{k-1}})=(i, j)} $ represents the invariant of the $i$-th loop execution at Level $k-1$, where $j$ is any iteration index ($1 \leq j \leq N_{k-1}$). Hence, for any $(m_k,n_k)$, we have
$$R(v^k_{q_k})[S^*(v^{k-1}_1), \cdots, S^*(v^{k-1}_{h_{k-1}})] |_{\mathcal{L}(v^k_{q_k})=(m_k,n_k)} = S^*(v^{k-1}_{p_{k-1}}) |_{\mathcal{L}(v^{k-1}_{p_{k-1}})=(((m_k-1) \cdot N_{k} + n_{k}) \cdot N_{k-1}, N_{k-1})}.$$
Consequently, we obtain 
$$S^*(v^k_{q_k}) |_{\mathcal{L}(v^k_{q_k})=(m_k,n_k)} =  R(v^k_{q_k}) [S^*(v^{k-1}_1), S^*(v^{k-1}_2), \cdots, S^*(v^{k-1}_{h_{k-1}})]|_{\mathcal{L}(v^k_{q_k})=(m_k,n_k)},$$
which means that $ S^*(v^k_{q_k}) = R(v^k_{q_k}) [S^*(v^{k-1}_1), S^*(v^{k-1}_2), \cdots, S^*(v^{k-1}_{h_{k-1}})]$. 

In conclusion, the summary configuration $(S^*(v^1_1), S^*(v^1_2), \cdots)$ is self-consists over $G^*(V^*,E^*)$.

\stitle{Step 2: without loss of generality, we just prove that Theorem \ref{theorem: recursive lifting} is valid for $K =2$.}

Based on Theorem \ref{theorem: self-consistent}, $S^*(v^{2}_{p_{2}})|_{\mathcal{L}(v^{2}_{p_{2}})=(m_{2}, j)}$ represents the invariant of the second layer loop at the $m_2$-th loop execution where $j$ is an iteration index ($1 \leq j \leq N_{2}$). Since Level $2$ is the top level, we have $M_2 = 1$. Due to $1 \leq m_2 \leq M_2 =1$, $m_2$ is always equal to $1$. Further, similar to the poof of Theorem \ref{theorem: self-consistent}, there exists a affine map $T$ such that $S^*(v^{2}_{p_{2}})$ can be rewritten as
$$S^*(v^{2}_{p_{2}})|_{\mathcal{L}(v^{2}_{p_{2}})=(1, j)} = T( S^*(v^{2}_{1})|_{\mathcal{L}(v^{2}_{1})=(1, j-1)}, \cdots, S^*(v^{2}_{q_2})|_{\mathcal{L}(v^{2}_{q_2})=(1, j-1)}, \cdots, S^*(v^{2}_{h^2})|_{\mathcal{L}(v^{2}_{h_2})=(1, j-1)}).$$
According to the feature of the vertex $v^{2}_{q_2}$, we have
\begin{equation}
S^*(v^{2}_{q_2})|_{\mathcal{L}(v^{2}_{q_2})=(1, j-1)} = S^*(v^{1}_{p_1})|_{\mathcal{L}(v^{1}_{p_1})=((j-1) \cdot N_1, N_1)}.
\end{equation}
Let
\begin{equation}
\label{equation: Vij}
\begin{cases}
\eta(i, j) = T( \cdots, S^*(v^{2}_{q_2-1})|_{\mathcal{L}(v^{2}_{q_2-1})=(1, j-1)}, ~ \theta(j-1,i), ~ S^*(v^{2}_{q_2+1})|_{\mathcal{L}(v^{2}_{q_2+1})=(1, j-1)}, \cdots ),  \\
\theta(j-1,i) = S^*(v^{1}_{p_1})|_{\mathcal{L}(v^{1}_{p_1})=((j-1) \cdot N_1, i)}.
\end{cases}
\end{equation}
In the following, we prove that $\eta(i, j)$ is an invariant for any iteration. First, if the iteration appears during the loop execution at Level 2, then the inner loop is executed once. Hence, we have
\begin{equation}
\begin{aligned}
Iter_{j} (\eta(1, j)) & = Iter_{j} T(S^*(v^{2}_{1})|_{\mathcal{L}(v^{2}_{1})=(1, j-1)}, \cdots, \theta(j-1,M_1), \cdots, S^*(v^{2}_{h_2})|_{\mathcal{L}(v^{2}_{h_2})=(1, j-1)}) \\
& = Iter_{j} (S^*(v^{2}_{p_{2}})|_{\mathcal{L}(v^{2}_{p_{2}})=(1, j)}) = (S^*(v^{2}_{p_{2}})|_{\mathcal{L}(v^{2}_{p_{2}})=(1, j+1)}) = \eta(1, j+1).
\end{aligned}
\end{equation}
Second, if the iteration appears in the loop execution at Level 1, then we have
\begin{displaymath}
Iter_{i} (\eta(i, j)) = T(S^*(v^{2}_{1})|_{\mathcal{L}(v^{2}_{1})=(1, j-1)}, \cdots, Iter_{i} (\theta(j-1,i)), \cdots, S^*(v^{2}_{h_2})|_{\mathcal{L}(v^{2}_{h_2})=(1, j-1)}).
\end{displaymath}
Since $S^*(v^{1}_{p_1})$ is invariant at Level 1, we get
$$Iter_{i} (\theta(j-1,i)) = Iter_{i} (S^*(v^{1}_{p_1})|_{\mathcal{L}(v^{1}_{p_1})=((j-1) \cdot N_1, i)}) = S^*(v^{1}_{p_1})|_{\mathcal{L}(v^{1}_{p_1})=((j-1) \cdot N_1, i+1)}.$$
Hence, $ Iter_{i} (\eta(i, j)) = \eta(i+1, j)$.

Thus, $\eta(i, j)$ is the invariant, and $\eta(M_2, N_2)$ yields the postcondition $S^*(v^2_{p_2})|_{\mathcal{L}(v^2_{p_2})=(1,N_2)}$.
\end{proof}

\section{Proofs of Theorem \ref{theorem: finite termination property} and Theorem \ref{theorem: hierarchical recursive lifting}}
\label{Asection: Proofs 2}

\begin{theorem}[Theorem \ref{theorem: finite termination property}] For any SCC in the invariant subgraph corresponding to a loop in a $d$-dimensional $m$-point stencil computation, there exists a positive integer $n_{iter}$ such that the iterative semantic extraction process converges after executing $n_{iter}$ sweeps.
\end{theorem}
\begin{proof}
To prove Theorem \ref{theorem: finite termination property}, we only need to show that the summary $S(v_{start})$ of the start vertex converges to an invariant form.

First, in an iteration of the semantic extraction process above, the number of expressions $N_e$ in $S(v_{start})$ increases, which implies that there are dependencies between the elements of the output matrix $B$, as illustrated in Figure \ref{fig:irreg-scc}. For an $d$-dimensional $m$-point stencil computation, $N_e$ could not grow indefinitely. We use proof by contradiction to establish this claim. Now, assume that $N_e$ continues to grow indefinitely with the iterative semantic extraction process. This would mean that we can identify $m'$ distinct positions $\mathbf{x}^{(1)}, \mathbf{x}^{(2)}, \cdots, \mathbf{x}^{(m')}$ in $B$, where $m' > m$, such that the element at position $\mathbf{x}^{(i)}$ depends on the element at position $\mathbf{x}^{(i+1)}$ for $ 1 \leq i \leq m'-1 $. Meanwhile, according to Equation (\ref{equation: stencil computation pattern}), there exist $m'$ distinct positions $\mathbf{y}^{(1)}, \mathbf{y}^{(2)}, \cdots, \mathbf{y}^{(m')}$ in the input array $A$ such that the computation of $B(\mathbf{x}^{(i)})$ requires $A(\mathbf{y}^{(i)})$. Moreover, since $B(\mathbf{x}^{(i)})$ depends on $B(\mathbf{x}^{(i+1)})$, the computation of $B(\mathbf{x}^{(1)})$ would require values from $ A(\mathbf{y}^{(1)}), A(\mathbf{y}^{(2)}), \cdots, A(\mathbf{y}^{(m')}) $, where $m' > m$. This contradicts the definition of an $m$-point stencil, proving that $N_e$ could not increase indefinitely.  

Next, once $N_e$ stabilizes during the iterative semantic extraction process, $Step~3$ of the semantic extraction process generalizes $S(v_{start})$, where each expression $\tilde{E}_k(\mathbf{x})$ becomes independent of the iteration variable $t$. Hence, $\tilde{E}_k$ also becomes invariant.

Furthermore, the validity condition $P_k$ of $\tilde{E}_k$ is transformed in \textit{Step 3} into the form: $LB^k_{i}(t) \leq x_i \leq UB^{k}_{i}(t)$. For the $l$-th invariant subgraph, $t$ satisfies $ 1 \leq t \leq N_l $, meaning $ LB^k_{i}(t) $ has a lower bound and $ UB^k_{i}(t) $ has an upper bound. Hence, after a finite number of sweeps, $P_k $ remains unchanged.

Consequently, we conclude that the iterative semantic extraction terminates in finite sweeps.
\end{proof}

\begin{theorem}[Theorem \ref{theorem: hierarchical recursive lifting}]
Given the invariant graph for a stencil, the hierarchical recursive lifting algorithm
possesses the finite termination property and yields the summary/postcondition of the stencil kernel.
\end{theorem}
\begin{proof}
According to Theorem \ref{theorem: finite termination property}, the iterative semantic extraction on each invariant subgraph has the finite termination property. Since the hierarchical recursive lifting consists of a finite number of recursive calls to iterative semantic extraction, from the lowermost level to the uppermost level, it also has the finite termination property. Furthermore, based on Theorem \ref{theorem: recursive lifting}, the \textit{self-consistent} summary configuration generated by Algorithm \ref{alg:recursive-semantic-lifting} leads to the postcondition of stencil kernel.
\end{proof}

\begin{figure}
    \centering
    \includegraphics[width=0.9\linewidth]{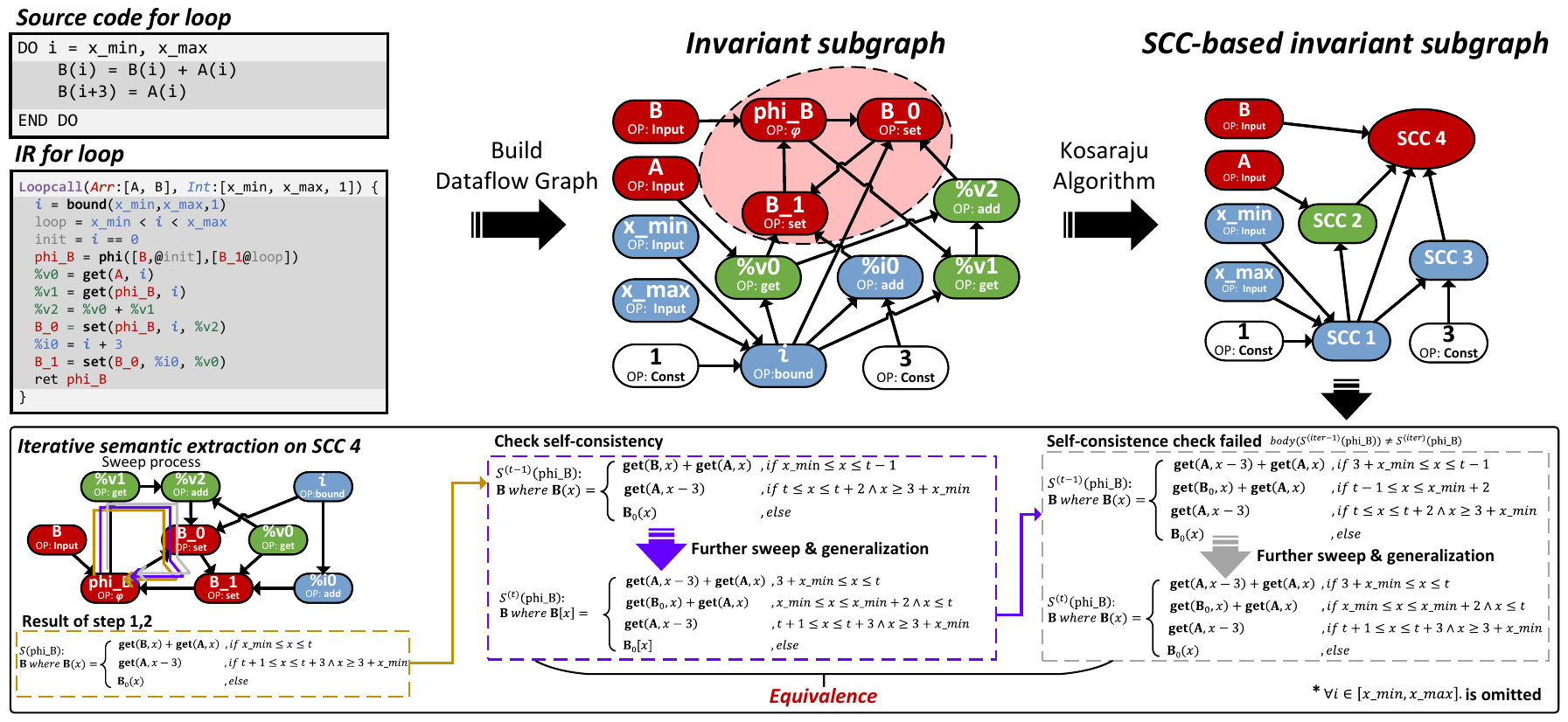}
    \vspace{-0.5em}
    \caption{Iterative semantic extraction within a non-trivial SCC for a loop with intermediate buffers.}
    \label{fig:irreg-scc}
\end{figure}

\begin{figure}
    \centering
    \includegraphics[width=0.9\linewidth]{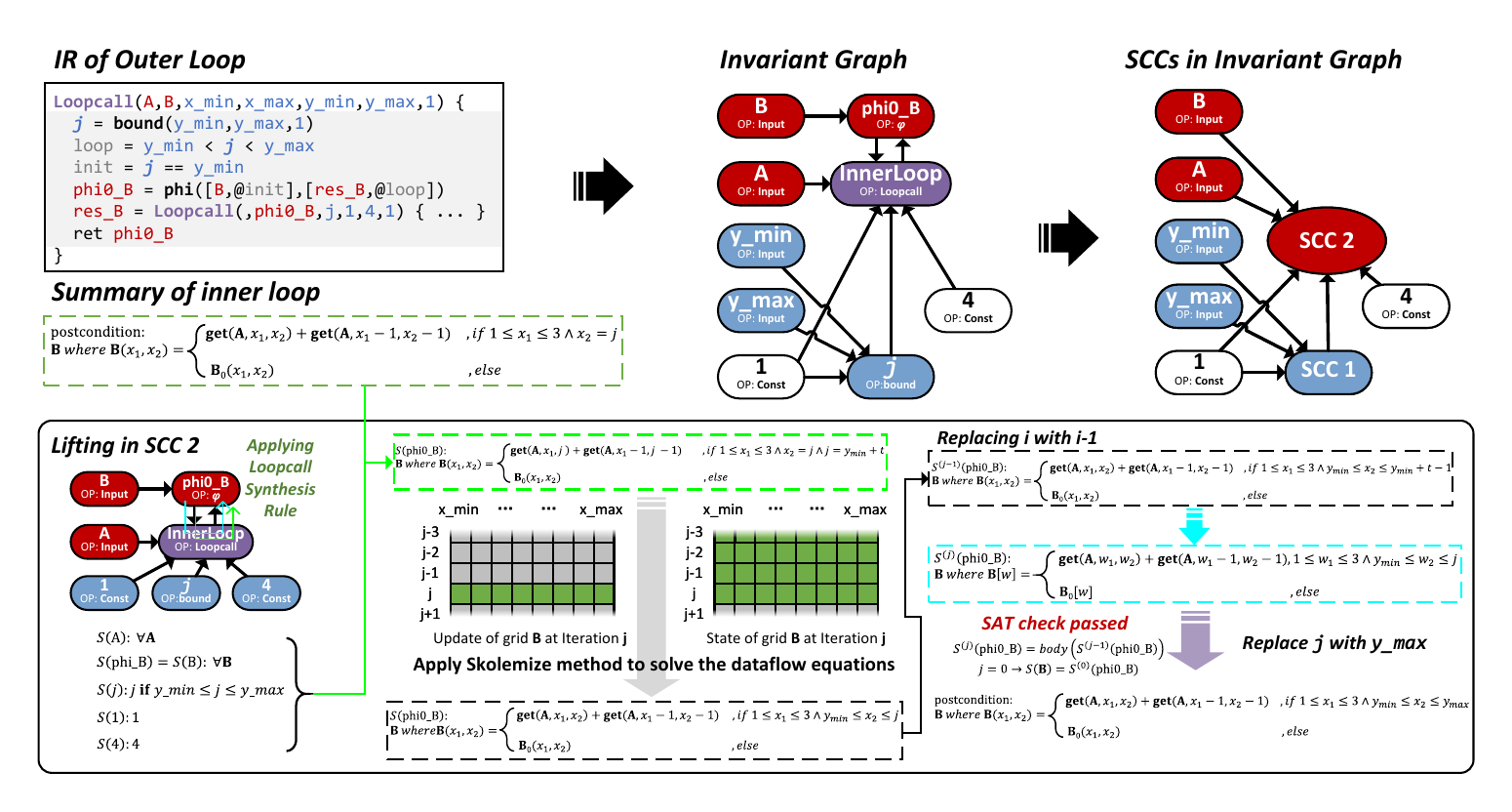}
    \vspace{-0.5em}
    \caption{Iterative semantic extraction on the invariant subgraph for the outer loop of stencil kernel in Figure \ref{fig: Basic idea of our design}.}
    \label{fig:regular-scc-out}
\end{figure}

\section{The example in Figure \ref{fig:regular-scc}}
\label{Asection: Hierarchical Recursive Lifting Algorithm}

We use the example in Figure \ref{fig:regular-scc} to show how the iterative semantic extraction works.

\stitle{\small \CIRCLE} \textbf{Step 1.} In Figure \ref{fig:regular-scc}, the operation $O(phi1\_B)$ is a $\varphi$ function, and the vertex $phi1\_B$ has a source vertex $B$, which is an input vertex whose \textit{self-consistent} summary has been determined to be $D_{\text{value}}(B)$, i.e., the array $\mathbf{B}_0$. According to \textit{Step 1}, the vertex $phi1\_B$ is selected as the start vertex, and its summary is initialized with the summary of $B$.

\stitle{\small \CIRCLE} \textbf{Step 2.} During a sweep in SCC~7, which consists of two vertices, \textit{phi1\_B} and \textit{B\_1}, the summary of \textit{B\_1} is updated first, followed by the summary update of the start vertex \textit{phi1\_B}. The resulting summary of \textit{phi1\_B} contains conditional branches that symbolically define how the fixed array element $B[x_1,x_2]$ is computed. Each branch corresponds to a particular control-flow path and reflects the influence of the summaries propagated from both internal and external source vertices.

\stitle{\small \CIRCLE} \textbf{Step 3.}  After sweep, we have $N_{e} = 1$ for the summary of \textit{phi1\_B}, $E_1(\mathbf{x},t):= \text{get}(\mathbf{A}, i,j) +\text{get}(\mathbf{A},i-1,j-1)$, and $P_1(\mathbf{x},t):=(\mathbf{x} = (x_1, x_2) = (i,j) \wedge t = i)$. By solving the corresponding predicate logic equation, we refine the expressions and their validity conditions. The expressions $E_1(\mathbf{x},t)$ is transformed into $\tilde{E}_1(\mathbf{x}):= \text{get}(\mathbf{A},x_1,x_2) +\text{get}(\mathbf{A},x_1-1,x_2-1)$, making it independent of the iteration variable $t$. The validity condition $P_1(\mathbf{x},t)$ becomes $\tilde{P}_1(\mathbf{x},t) := (\mathbf{x} = (x_1,x_2) \wedge 1\leq x_1 \leq t \wedge x_2 = j)$. Applying $\tilde{E}_1(\mathbf{x})$ and $\tilde{P}_1(\mathbf{x},t)$, we complete the construction of $S^{general}(phi1\_B)$.

\stitle{\small \CIRCLE} \textbf{Step 4.}
After performing the variable substitution in the generalized summary of \textit{phi1\_B}, a further sweep begins from \textit{phi1\_B} to derive a new summary.

\stitle{\small \CIRCLE} \textbf{Step 5.}
After a further sweep, $S^{(t)}(phi1\_B) = S^{general}(phi1\_B)$ is valid, confirming convergence. Thus, the summaries of all vertices within the current SCC become \textit{self-consistent}.

\etitle{Remark.} Figure \ref{fig:regular-scc-out} shows the iterative semantic extraction on the invariant subgraph for the outer loop of the stencil case in Figure \ref{fig: Basic idea of our design}.

\bibliographystyle{ACM-Reference-Format}
\bibliography{StencilLifting}

\end{document}